\documentclass[11pt]{article}
\pdfoutput=1
\usepackage[utf8]{inputenc}

\usepackage{hyperref}
\usepackage[boxruled]{algorithm2e}

\usepackage{fullpage}
\usepackage{amsmath}
\usepackage{amsfonts}
\usepackage{amssymb}
\usepackage{tikz}
\usepackage{mathrsfs}
\usepackage[colorinlistoftodos]{todonotes}

\usepackage{amsthm}
\usepackage{cases}
\newtheorem{theorem}{Theorem}
\newtheorem{lemma}[theorem]{Lemma}
\newtheorem{corollary}[theorem]{Corollary}

\newtheorem{claim}[theorem]{Claim}

\newtheorem{proposition}[theorem]{Proposition}

\newtheorem*{question}{Question}

\newtheorem{fact}[theorem]{Fact}
\newtheoremstyle{restate}{}{}{\itshape}{}{\bfseries}{~(restated).}{.5em}{\thmnote{#3}}
\theoremstyle{restate}
\newtheorem*{restate}{}

\newcommand{\E}{\mathop{\mathbb{E}}}

\newcommand{\HammingCube}[1]{\{0,1\}^{#1}}

\DeclareMathOperator*{\argmin}{arg\,min}

\newcommand{\mihai}{Mihai P\v{a}tra\c{s}cu}

\newcommand{\pub}{\mathrm{pub}}
\newcommand{\probe}{\mathtt{Probe}}
\newcommand{\pred}{\mathtt{pred}}
\newcommand{\com}{\mathrm{ext}}
\newcommand{\poly}{\mathrm{poly}}
\newcommand{\rmq}{\texttt{RMQ}}
\newcommand{\cQ}{\mathcal{Q}}
\newcommand{\predz}{\texttt{pred-z}}
\newcommand{\pt}{\mathrm{pt}}
\newcommand{\cS}{\mathcal{S}}
\newcommand{\cE}{\mathcal{E}}
\newcommand{\good}{\mathrm{good}}
\newcommand{\foot}{\mathtt{Foot}}
\def\mainfile{}

\begin{document}
\title{Lower Bound for Succinct Range Minimum Query}
\author{Mingmou Liu\thanks{Nanjing University. \texttt{liu.mingmou@smail.nju.edu.cn}. Supported by National Key R\&D Program of China 2018YFB1003202 and the National Science Foundation of China under Grant Nos. 61722207 and 61672275. Part of the research was done when Mingmou Liu was visiting the Harvard University.}\and Huacheng Yu\thanks{Princeton University. \texttt{yuhch123@gmail.com}.}}
\date{\today}
\maketitle

\setcounter{page}{0}
\maketitle
\thispagestyle{empty}

\begin{abstract}
	Given an integer array $A[1..n]$, the Range Minimum Query problem (RMQ) asks to preprocess $A$ into a data structure, supporting RMQ queries: given $a,b\in [1,n]$, return the index $i\in[a,b]$ that minimizes $A[i]$, i.e., $\argmin_{i\in[a,b]} A[i]$.
	This problem has a classic solution using $O(n)$ space and $O(1)$ query time by Gabow, Bentley, Tarjan~\cite{GBT84} and Harel, Tarjan~\cite{HT84}.
	The best known data structure by Fischer, Heun~\cite{FH11} and Navarro, Sadakane~\cite{NS14} uses $2n+n/(\frac{\log n}{t})^t+\tilde{O}(n^{3/4})$ bits and answers queries in $O(t)$ time, assuming the word-size is $w=\Theta(\log n)$.
	In particular, it uses $2n+n/\poly\log n$ bits of space when the query time is a constant.

	In this paper, we prove the first lower bound for this problem, showing that $2n+n/\poly\log n$ space is necessary for constant query time.
	In general, we show that if the data structure has query time $O(t)$, then it must use at least $2n+n/(\log n)^{\tilde{O}(t^2)}$ space, in the cell-probe model with word-size $w=\Theta(\log n)$.
\end{abstract}

\newpage

\ifx\mainfile\undefined
\documentclass[11pt]{article}

\begin{document}
\fi

\section{Introduction}
Given an array $A[1..n]$ of integers, the Range Minimum Query (RMQ) problem asks to preprocess $A$ into a data structure, supporting 
\begin{itemize}
	\item \rmq($a$, $b$): return $\argmin_{i\in[a,b]} A[i]$ (if multiple entries have the smallest $A[i]$, return the one with smallest $i$).
\end{itemize}
RMQ data structures have numerous applications in computer science, for instance, in text processing~\cite{ALV92,Mut02,FHK06,Sada07b,Sada07,VM07,CPS08,FMN09,HSV09,CIKRTW12}, graph problems~\cite{RV88,BV93,GT04,BFPSS05,LC08} and other areas of computer science~\cite{Sax09,SK03,CC07}.

The RMQ problem has a classic ``textbook'' solution using $O(n)$ space and answering each \rmq{} query in constant time, due to Gabow, Bentley, Tarjan~\cite{GBT84} and Harel, Tarjan~\cite{HT84}.
To build this data structure, one first constructs the \emph{Cartesian tree} of the input array $A$.
The Cartesian tree is an $n$-node binary tree, where the nodes correspond to the entries of $A$.
The root $A[i]$ has the minimum value in $A$, and its left and right subtrees are recursively constructed on $A[1..i-1]$ and $A[i+1..n]$ respectively.
It turns out that \rmq($a$, $b$) is exactly the lowest common ancestor (LCA) of nodes $A[a]$ and $A[b]$ in this tree.
The LCA problem admits an $O(n)$ space and $O(1)$ query time solution, based on a reduction to the $\pm 1$RMQ problem.

In terms of the space usage, this data structure is in fact suboptimal.
The only information needed to answer all \rmq{} queries on $A$ is the Cartesian tree of $A$, which is a rooted binary tree with $n$ nodes.
It is well known that the number of such binary trees is equal to the $n$-th Catalan number $C_n=\frac{1}{n+1}\binom{2n}{n}$.
Hence, the information theoretical space lower bound for this problem is in fact, $\log_2 C_n=2n-\Theta(\log n)$ \emph{bits}, whereas the above data structure uses $O(n)$ \emph{words} of space.

Sadakane~\cite{Sada07} showed that it is possible to achieve ``truly'' linear space and constant query time.
He proposed a data structure using $\sim 4n$ bits of space and supporting \rmq{} queries in $O(1)$ time, assuming the word-size is $\Omega(\log n)$.\footnote{It is a standard assumption that the word-size is at least $\Omega(\log n)$, since the query answer is an index, which requires $\log n$ bits to describe.}
Later, the space is further improved to $2n+O(n\log\log n/\log n)$ by Fischer and Heun~\cite{FH07,FH11}.
The space bound matches the best possible in the leading constant.
Such data structures that use $H+r$ bits, for problems that require $H$ bits of space\footnote{We denote the information theoretical minimum space by $H$, since the minimum space is exactly the entropy of the truth table for random database.} and for $r=o(H)$, are called the \emph{succinct} data structures~\cite{Jacobson:1988}.
The amount of the extra space $r$ is usually referred to as the \emph{redundancy}.
For succinct data structures, the main focus is on the trade-off between the redundancy and query time~\cite{GM07,Pat08,PV10}.

The state-of-the-art\footnote{Due to the wide application, special cases of the problem are also interesting.
	Researchers have designed data structures that outperform the best worst-case solution on inputs with certain structures~\cite{ACN13,BLRSSW15,GJMW19}.
} RMQ data structure~\cite{Pat08,NS14,FH11,DRS17} uses $2n+n/(\frac{\log n}{t})^{\Omega(t)}+\tilde{O}(n^{3/4})$ bits of space and answers queries in $O(t)$ time, for any parameter $t>1$ and word-size $\Theta(\log n)$.\footnote{Their data structure was originally stated as $2n+n/\log^t n$ space and $O(t)$ time for constant $t$.
It naturally generalizes to the above trade-off.
See Appendix~\ref{app_trade_off}.
}
In particular, one can achieve $n/\poly\log n$ redundancy and constant query time for any $\poly\log n$.
On the other hand, despite the fact that the Range Minimum Query problem has drawn a significant amount of attention from the algorithm community, to the best of our knowledge, no lower bound is known.
\begin{question}
	What is the lowest possible redundancy for RMQ data structures with constant query time?
\end{question}

\newcommand{\thmmaincont}{
	Given an array $A[1..n]$, for any data structure supporting \rmq{} queries using $2n+r$ bits of space and query time $t$, we must have
	\[
		r\geq n/w^{O(t^2\log^2 t)},
	\]
	in the cell-probe model with word-size $w\geq \Omega(\log n)$.
}

\paragraph{Our contribution.} 
In this paper, we prove the first lower bound on the trade-off between redundancy and query time for the Range Minimum Query problem.
In particular, we prove that for constant query time, one must use $n/\poly\log n$ bits of redundancy, answering the above question.
Our lower bound also extends to the following full trade-off between the redundancy and query time.
\begin{theorem}\label{thm_main}
	\thmmaincont
\end{theorem}

Our proof technique is inspired by a lower bound of P\v{a}tra\c{s}cu and Viola~\cite{PV10}, which asserts a similar trade-off for a different data structure problem (see Section~\ref{sec_pv}).
Nevertheless, our proof majorly deviates from \cite{PV10}, due to the nature of the differences in the two problems.
In the following, we briefly survey the known techniques for proving redundancy-query time trade-off, and the difficulties in generalizing the proofs to our problem.
In Section~\ref{sec_overview}, we present a technical overview of our proof.
Finally, we prove our main theorem in Section~\ref{sec_lower} and Section~\ref{sec_dist}.
In Appendix~\ref{app_upper_one} and Appendix~\ref{app_trade_off}, we present a simplified upper bound for \rmq{}.

\subsection{The cell-probe model}
The cell-probe model~\cite{yao1981should} is a classic computational model for studying the complexity of data structures, which is similar to the well-known RAM model, except that all operations are free except memory access.
It is easy to see the practicality of the theoretical model: the time cost on memory accessing is dominant over the time costs on all other operations in a typical modern computer.
In particular, a data structure that occupies $s$ memory units on a machine with word size of $w$ bits and answers queries by accessing at most $t$ memory units is described by a preprocessing algorithm and a query algorithm.
Given a database $D$, the preprocessing algorithm outputs a table consists $s$ cells, $w$ bits per cell, which stands for our data structure.
We emphasize that the preprocessing algorithm can take arbitrarily long time and arbitrarily large space, but has to halt and output the table.
On the other hand, the query algorithm captures the interaction between CPU and data structure when the CPU is trying to answer a query.
Given a query $q$, the query algorithm makes cell-probes (i.e. memory accesses) \emph{adaptively} in following sense:
the algorithm reads the query $q$, then chooses a cell $c_1$, makes a cell-probe to retrieve the content of $c_1$;
combining $q$ and content of cell $c_1$, the algorithm chooses second cell $c_2$, makes a cell-probe to retrieve the content of cell $c_2$; and goes on; after retrieved the content of cell $c_t$, the algorithm outputs the answer by combining $q$ and the contents of cells $c_1,\dots,c_t$.
There is a substantial body of notable works on the cell-probe complexity of static data structure problems~\cite{miltersen1998data,chakrabarti2010optimal,patrascu2006higher,panigrahy2008geometric,PTW10,Larsen12b,Yin16}.

\subsection{Related work}
One of the most widely used technique in proving data structure lower bounds is the \emph{cell-sampling}.
It has applications to dynamic data structures~\cite{Larsen12a,CGL15,LWY18}, streaming lower bounds~\cite{LNN15}, static data structures with low space usage~\cite{PTW10,Larsen12b,GL16,Yin16}, as well as succinct data structures~\cite{GM07}.
To prove a lower bound using cell-sampling, one samples a small fraction of the memory cells, and argues that if the query time is low, then many queries can still be answered using only the sampled cells.
Finally, one proves that too much information about the input is revealed by the answers to those queries, yielding a contradiction.
In order to apply this technique, the problem is usually required to have the property that the answers to a set of $n^{1.1}$ random queries almost determine the entire input data, which the RMQ problem does not have (as one could keep getting indices with small values).

Golynski~\cite{Gol09} developed a different technique for proving succinct data structure lower bounds, and proved lower bounds for several problems, including storing a permutation $\pi$ supporting $\pi,\pi^{-1}$ queries, and storing a string $S$ supporting pattern matching (return the $i$-th occurrence of $P$) and substring access queries (return the substring $S[i..i+p]$).
His technique mostly applies to problems with two types of queries that ``verifies'' each other, e.g., $\pi(i)=j$ and $\pi^{-1}(j)=i$.

\subsection{P\v{a}tra\c{s}cu and Viola's technique}\label{sec_pv}
Our proof uses a few important ideas from~\cite{PV10}.
P\v{a}tra\c{s}cu and Viola proved a lower bound for succinct rank data structures.
The rank problem asks to preprocess a 0-1 array $A[1..n]$ into a data structure, supporting queries of form ``return the number of ones in $A[1..i]$.''
They proved that in the cell-probe model with word-size $w\geq \Omega(\log n)$, if a data structure uses at most $n+n/w^t$ bits of space, then its query time must be at least $\Omega(t)$, which is known to be tight in the cell-probe model~\cite{Pat08,Yu19}.

P\v{a}tra\c{s}cu and Viola's proof uses a variant of \emph{round elimination}.
We first fix the input distribution to be uniform. 
In each round of the argument, we are given a cell-probe data structure that uses $n$ bits of memory with additional $p$ \emph{published bits}, where the published bits can be accessed by the query algorithm at no cost.
In the other words, the data structure uses $p$ more bits than the information theoretical lower bound, and the $p$ extra bits are given to the query algorithm for free.
Then, we create a new data structure with a factor of $w^{O(1)}$ more publish bits, and at the same time, argue that its expected query time must decrease by some constant $\epsilon$, where the expectation is taken over a random input and a random query.
To see why it already implies the above trade-off, suppose we have a cell-probe data structure with $n+n/w^t$ bits of space, we first ``publish'' the last $n/w^t$ bits, and then apply the above argument for $\Theta(t)$ rounds.
Thereafter, a total of strictly less than $n$ bits are published, while the query time is decreased by $\Theta(t)$.
If the initial query time was much less than $t$, then we would have obtained a data structure with $<n$ published bits and query time $0$, i.e., all queries can be answered by only reading the published bits, yielding a contradiction.

The key argument lies in choosing the extra bits to publish and proving the decay of the query time.
To this end, consider two sets of queries $Q_1$ and $Q_2$ of size $O(p)$, comparable to the number of published bits.
We can show that for the rank problem, if both $Q_1$ and $Q_2$ are evenly distributed over the $n$ possible queries, then their answers are very correlated.
In particular, the mutual information between the answers to $Q_1$ and $Q_2$ is at least $\Omega(p)$.
Roughly speaking, there are $\Omega(p)$ bits of information about the input array that is both contained in $Q_1$ and $Q_2$.
It implies that in order to answer $Q_2$, the query algorithm must probe many cells that are also probed when answering $Q_1$.
This is because otherwise, most of the $\Omega(p)$ bits of ``shared information'' between the answers would be stored in two separate locations in the memory, ``wasting'' $\Omega(p)$ bits of space.
It is unaffordable, as the data structure uses only $p$ extra bits.

The above argument shows that if we publish all memory cells that are probed when answering $Q_1$ (by publishing the bits encoding their addresses and contents), then for an average query $q\in Q_2$, $\Omega(1)$ cells that are probed when answering $q$ get published, i.e., the query algorithm can now access their contents for free (by reading the published bits), and the query time is reduced by $\Omega(1)$.
Moreover, this argument works for the \emph{same} $Q_1$ and \emph{every} evenly distributed $Q_2$. Therefore, by publishing those cells, the expected query time of a random query must decrease by a constant.
On the other hand, the total number of published bits is at most $O(|Q_1|\cdot tw)\leq p\cdot w^{O(1)}$.
This completes the argument for one round, and by the earlier argument, it implies the lower bound.

\bigskip

The main obstacle in applying this strategy to the RMQ problem is to prove the correlation lower bound between two sets of queries $Q_1$ and $Q_2$.
More precisely, one needs to prove that for two \emph{random} sets of $p$ queries $Q_1$ and $Q_2$, the mutual information between their answers (assuming uniformly random input data) is at least $\Omega(p)$.
Unfortunatly, this is not true for the RMQ problem, because even the entropy of the answers to $p$ random queries is significantly lower than $\Omega(p)$ -- a random query has $\Omega(n)$ length and the small entries are likely to appear in many answers.
One simple idea to resolve this particular issue is to only consider shorter queries, and prove lower bounds on their query time.
For instance, if we only consider queries of length $O(n/p)$, then two random sets $Q_1$ and $Q_2$ of $p$ queries will both be spread out and significantly overlap.
In this case, one can show their mutual information is indeed $\Omega(p)$.
However, our argument proceeds in rounds, and the value of $p$ increases by a factor of $w^{O(1)}$ each round.
We must consider the \emph{same} set of queries and reduce their query time in all rounds.
This simple ``hack'' of the proof does not solve the problem.

\ifx\mainfile\undefined
\bibliography{refs}
\bibliographystyle{alpha}

\end{document}
\fi

\ifx\mainfile\undefined
\documentclass{article}

\begin{document}
\fi

\section{Our Technique}\label{sec_overview}

In the previous section, we showed a concrete technical difficulty to apply~\cite{PV10} directly to our problem.
It turns out that the more inherent reason is that their technique is too ``strong'': It lower bounds the expected query time of a query on a \emph{random} input database, i.e., proving an \emph{average-case} lower bound.
For the RMQ problem, we believe much more efficient average-case solutions exist (see also the next paragraph).
In this case, the techniques from \cite{PV10} would become inapplicable for proving a high lower bound.
This observation also suggests that the hard queries should be chosen \emph{depending on the input database}.
This is one major modification in our argument, which turns out to cause new issues.
We will elaborate below.

\bigskip
To prove the RMQ lower bound, we first reduce RMQ from a variant of the \emph{predecessor search} problem.
In the predecessor search problem, we are given a sorted list $S=\{s_1,s_2,\cdots,s_m\}\subset [U]$, and asked to preprocess $S$ into a data structure so that given a query $x\in[U]$, the largest element in $S$ that is at most $x$ can be found efficiently.
To see why RMQ is even related to predecessor search, let us consider all queries of form $\rmq(n/2, x)$.
The only indices $i$ that could become an answer to (at least) one of such queries are the ones with a smaller value than all other entries in $A[n/2..i-1]$.
Denote this set by $S$, then the answer to $\rmq(n/2, x)$ is precisely the predecessor of $x$ in $S$.
A more careful analysis shows that for a uniformly random Cartesian tree, $|S|$ is likely to be $\Theta(\sqrt{n})$.\footnote{This might be counter-intuitive at the first glance, as a random array $A$ would only generate an $S$ of size $\Theta(\log n)$. However, note that the distribution of the Cartesian tree generated by a uniformly random $A$ is in fact, (very) different from \emph{a uniformly random Cartesian tree}. The space benchmark of the RMQ problem is based on the number of different $n$-node Cartesian trees $C_n$. Therefore, sampling a uniform Cartesian tree would maximize the input entropy (matching the space benchmark), and is the right input distribution to keep in mind.
}
A classic lower bound for predecessor search~\cite{patracscu2006time} shows that for such instances, any data structure must use at least $\Omega(\log\log n)$ query time, even with linear space (and not in the succinct regime).
However, the distribution of $S$ induced by a random Cartesian tree is different from the hard instances from~\cite{patracscu2006time}, which makes the instances easier and relevant for the succinct regime, and at the same time, it requires a different argument to prove lower bounds.
For now, let us think of the space benchmark (the information theoretical lower bound) being the entropy of the input $H(S)$.
Note that since $S$ is non-uniform, this space benchmark could only be achieved in expectation.
See the next section for the formal definition of the problem, which allows us to define the space benchmark in worst case.
It turns out that this variant of predecessor search has a cell-probe data structure with constant redundancy and constant average-case query time.
Therefore, as mentioned in the previous paragraph, in order to prove any non-trivial lower bound, we must choose queries based on the data.

To prove the lower bound, we first observe that this average-case data structure has a very slow query time on the \emph{exact input data points} (all points in $S$).
Thus, we will define the set of all predecessor search queries $\pred(x)$ for $x\in S$ (which returns  $x$) to be our hard queries $\cQ$.
The high-level strategy is similar to \cite{PV10}, using \emph{round elimination}:
Given a data structure using optimal space ``$H(S)$'' bits with $p$ extra \emph{published bits}, we find a set of queries $Q_{\pub}$ of size $\tilde{O}(p)$, and prove that an average query in $\cQ$ must probe $\epsilon$ cells that are also probed by $Q_{\pub}$ in expectation; then we publish all the cells probed by $Q_{\pub}$, which reduces the average query time of $\cQ$ by $\epsilon$.
In each round, we publish a factor of $\poly\log n$ more bits while reducing the average query time of $\cQ$ by $\epsilon$.
On the other hand, we must have published at least $\tilde{\Omega}(|S|)$ bits before reducing the query time to $0$, which implies a lower bound on the initial average query time of $\cQ$.
Hence, same as \cite{PV10}, the key argument lies in finding such a set $Q_{\pub}$ (and proving the decay of the query time).

To this end, let us fix a set $Q_{\pub}$, whose answers have entropy much higher than $p$.
The goal is to show $\epsilon$-fraction of the queries in $\cQ$ probe cells that are also probed by $Q_{\pub}$.
It is actually not hard to prove a weaker statement: \emph{at least one} query in $\cQ$ probe cells also probed by $Q_{\pub}$.
If the sets of cells probed by $\cQ$ and $Q_{\pub}$ were disjoint, then by \emph{deleting} all cells probed by $Q_{\pub}$, we would obtain a significantly smaller data structure that still encodes the whole database.
To recover the database, we go over all queries and try to answer each of them without using the deleted cells.
The whole database can be recovered, since all $\cQ$ can be answered, and their answers determine $S$.
On the other hand, since $Q_{\pub}$'s answers have entropy much higher than $p$, it means that we must have deleted much more than $p$ bits from the data structure, yielding a contradiction.\footnote{The status of a cell being ``deleted'' still carries information, but as long as $w\geq 2\log n$, this is not an issue.}
However, this argument does not extend to proving $\epsilon$-fraction of $\cQ$ must probe the cells probed by $Q_\pub$.
Since assume for contradiction this is not the case, then after deleting all cells probed by $Q_{\pub}$, the data structure can recover only up to $(1-\epsilon)$-fraction of the data points (the ones that do not use cells probed by $Q_\pub$), there will be no contradiction unless the entropy of $Q_\pub$ is at least $\epsilon H(S)$, which can be significantly larger than $p$.

However, this $(1-\epsilon)$-fraction of $\cQ$ may reveal a lot of information about $Q_{\pub}$, as the entire $\cQ$ completely determines the answers to $Q_{\pub}$.
Intuitively, we should be able to \emph{compress} the set of the cells probed by $Q_{\pub}$, \emph{given} the set of all the other cells, since the former determines the answers to $Q_{\pub}$, the latter determines the answers to this $(1-\epsilon)$-fraction of $\cQ$, and they must have high mutual information.
That means we could compress the whole data structure by first encoding the set of all other cells, then writing down the ``compressed'' encoding of set of cells probed by $Q_{\pub}$.
Such compression would yield a contradiction.
However, there is a very subtle issue in implementing this idea: it is possible that this mutual information comes from the \emph{addresses} of the two sets, but not their contents.
That is, the queries in $Q_{\pub}$ can be adaptive, and both sets of cells determine \emph{which cells} are probed by $Q_{\pub}$, which contains information.
For example, this will not be an issue when the query algorithm is \emph{non-adaptive}, i.e. the set of probed cells depends only on the query, but not the database.
In this case, the addresses of the cells probed by $Q_{\pub}$ are fixed, and they have no information about the database.
To obtain a contradiction, we compress the memory of the data structure as follows: write down all published bits and the contents of all cells that are not probed by $Q_{\pub}$, then encode the cells probed by $Q_{\pub}$ conditioned on the cells we have written down.
From the first part of the compression, one can recover the $(1-\epsilon)$-fraction of $\cQ$ that does not use cells probed by $Q_{\pub}$ (since the data structure is non-adaptive, the decoding algorithm knows which cells are not probed by $Q_{\pub}$ and are encoded here).
If they reveal more than $p$ bits of information about $Q_{\pub}$, then the last part of the compression saves more than $p$ bits, yielding a contradiction.

For general data structures, which can be adaptive, observe that the addresses of the $\ell$-th probe of $Q_\pub$ are determined by the contents of the previous $\ell-1$ probes.
We will choose an $\ell$ such that the $(1-\epsilon)$-fraction of $\cQ$ reveals sufficient information about \emph{the $\ell$-th probe} of $Q_{\pub}$, conditioned on the first $\ell-1$ probes.
By the chain rule of mutual information, such $\ell$ exists.
At the same time, by conditioning on the first $\ell-1$ probes, the addresses of the $\ell$-th probes no longer carry information.
The final argument is similar to the non-adaptive case:
write down the $p$ published bits, the contents of the first $(\ell-1)$ probes and the contents of all cells except the $\ell$-th probes; at last, encode the contents of the $\ell$-th probes conditioned on the cells we have written down.
One can show that this encoding compresses the data structure below the input entropy if $\epsilon$ is too small, which implies a lower bound on how much the average query time of $\cQ$ must decrease in each round, and in turn, it implies a query time lower bound.
See Section~\ref{sec_qpub} for the detailed argument.

\ifx\mainfile\undefined
\end{document}
\fi

\ifx\mainfile\undefined
\documentclass[11pt]{article}

\begin{document}
\newcommand{\thmmaincont}{
	Given an array $A[1..n]$, for any data structure supporting \rmq{} queries using $2n+r$ bits of space and query time $t$, we must have
	\[
		r\geq n/w^{O(t^2\log^2 t)},
	\]
	in the cell-probe model with word-size $w\geq \Omega(\log n)$.
}
\fi
\def\lbfile{}

\section{Lower Bound for Succinct Range Minimum Query}\label{sec_lower}
In this section, we prove our main theorem, a lower bound for succinct RMQ.
\begin{restate}[Theorem~\ref{thm_main}]
	\thmmaincont
\end{restate}

To prove the lower bound, we will reduce RMQ from a variant of the \emph{predecessor search} problem, which we refer to as $\predz$.
In this problem, we are given $d$ sets $S_1,\ldots,S_d\subseteq [B]$ of size $u$ for $u=\Theta(\sqrt{B})$,\footnote{$[B]:=\{1,...,B\}$.} together with a positive integer $z$, s.t.
\[
	1\leq z\leq Z\cdot \prod_{i=1}^d\prod_{j=0}^u C_{s^{(i)}_{j+1}-s^{(i)}_j-1},
\]
where $S_i=\{s_1^{(i)},\ldots,s_u^{(i)}\}$ such that $s^{(i)}_j<s^{(i)}_{j+1}$, $s_0^{(i)}$ is assumed to be $0$ and $s_{u+1}^{(i)}$ is assumed to be $B+1$, and $C_x=\frac{1}{x+1}\binom{2x}{x}$ is the $x$-th Catalan number.
Note that, $d,B,u$ and $Z$ are all parameters of the problem, the only inputs are the $d$ sets and the integer $z$.
The goal is to construct a data structure to store the sets and the integer, supporting
\begin{itemize}
	\item $\pred(i,x)$: return largest element in $S_i$ that is at most $x$, and if no such element exists, return $0$;
	\item \texttt{query-z}(): return $z$.
\end{itemize}
We are interested in two parameters of the data structure: the space usage and the query time of $\pred$ (and \texttt{query-z} could take arbitrarily long time).

The reason we involve the $z$ is that we would like to use a solution for RMQ as a black-box to solve the predecessor search.
However our version of predecessor search has low entropy, comparing with RMQ.
It turns out that the reduction will introduce a large redundancy for predecessor search, which makes it impossible to prove any non-trivial lower bound.
To avoid this, we move substantial information (i.e. the $z$) of the reduced RMQ instance to our predecessor search.
The reduced RMQ instance becomes a function of of $z$ and $S$'s. 
We then maximize the entropy of reduced RMQ instance by properly choosing the joint distribution of $Z$ and $S$'s, so that the reduction introduces at most $O(d\log B)$ bits of redundancy. 
See the proof of Theorem~\ref{thm_main} for more details.

One useful way to view the role of the integer $z$ is that if we sample a uniformly random input among all possible inputs, the existence of $z$ distorts the distribution of $\{S_i\}$.
More specifically, all $S_i$ will be mutually independent, while for each $S_i$, 
\[
	\Pr[S_i=\{s_1,\ldots,s_u\}]\propto \prod_{j=0}^u C_{s_{j+1}-s_j-1}.
\]
We will analyze this distribution more carefully in Section~\ref{sec_dist}.

Since the set of all queries can recover the sets and the integer, the data structure must store the entire input (and this is the purpose of having \texttt{query-z}).
The total number of possible inputs to the problem is
\[
	Z\cdot \left(\sum_{0<s_1<\cdots<s_u<B+1}\prod_{j=0}^uC_{s_{j+1}-s_j-1}\right)^d.
\]
By the Catalan $u$-fold convolution formula, it is equal to
\[
	Z\cdot \left(\frac{u+1}{2B-u+1}\binom{2B-u+1}{B+1}\right)^d.
\]
Therefore, by the fact that $u=\Theta(\sqrt{B})$, we obtain the following lemma on the information theoretical minimum space for this problem.
\begin{lemma}\label{lem_input}
	The information theoretical minimum space for \predz{} is
	\[
		H_{d,u,B,Z}:=d\cdot(2B-u-\Theta(\log B))+\log Z
	\]
	bits.
\end{lemma}

In the following subsections, we will prove the following lower bound for \predz{}.
\begin{lemma}\label{lem_predz}
	For any parameters $d,u,B$ and $Z$ satisfying $u=\Theta(\sqrt{B})$, any data structure for the \predz{} problem that uses at most $H_{d,u,B,Z}+O(d\log B)$ bits of space and answers pred queries in time $t$ must have
	\[
		(wt\log B)^{O(t^2\log^2 t)}\geq B,
	\]
	in the cell-probe model with word-size $w$.
\end{lemma}

Before we proceed and prove the lemma, let us first show that it implies our main theorem.

\begin{proof}[Proof of Theorem~\ref{thm_main}]
Suppose we have an RMQ data structure using space $2n+r$ bits with query time $O(t)$.
We will use it to solve \predz{} for
\begin{align*}
    \begin{cases}
    d:=2r\\
    B:=\left\lfloor \frac{n}{d}\right\rfloor-1\\
    u:=\lfloor \sqrt{B}\rfloor \\
    Z:=\binom{2u}{u}^r.
    \end{cases}
\end{align*}

Roughly speaking, we will divide the array $A$ into $d=2r$ blocks of length $B$ each (with a gap of one entry between the adjacent blocks), and embed one set $S_i$ into each block, e.g., $S_2$ is embedded into $A[B+2,\dots, 2B+1]$.
Suppose $S_2=\{s^{(2)}_1,\ldots,s^{(2)}_u\}$, then for each integer $j\in[u]$, the element $s^{(2)}_j$ corresponds to the entry $A[B+1+s_j]$.
We will ensure that $$A[B+1]>A[B+1+s^{(2)}_1]>\cdots>A[B+1+s^{(2)}_u],$$ i.e., the entries corresponding to elements in $\{0\}\cup S_2$ have decreasing values, and they are all smaller than other entries in $A[B+1,\dots,2B+1]$.
Hence, to answer the query $\pred(2, x)$, it suffices to send the query \rmq($B+1$, $B+1+x$) to the RMQ data structure, and the predecessor of $x$ is equal to $\rmq(B+1, B+1+x)-(B+1)$.
All even $S_i$ are embedded likewise, and for technical reasons, the odd $S_i$ are embedded with the universe reversed.
Finally, $z$ is used to encode other parts of $A$.
We elaborate below.

\bigskip

Given the inputs $\{S_1,\ldots,S_d\}$ and $z$ to the \predz{} problem, we have
\[
	z\leq \binom{2u}{u}^r\cdot \prod_{i=1}^{d}\prod_{j=0}^u C_{s^{(i)}_{j+1}-s^{(i)}_j-1},
\]
where $S_i=\{s^{(i)}_1,\ldots,s^{(i)}_u\}$.
We interpret $z$ as a tuple of $r+d(u+1)$ integers: $k_1,\ldots,k_r$ and $\{z^{(i)}_j\}$ for $i=1,\ldots,d$ and $j=0,\ldots,u$, such that $k_i\le\binom{2u}{u}$ and $z^{(i)}_j\le C_{s^{(i)}_{j+1}-s^{(i)}_j-1}$.
Now we construct the input $A$ to the RMQ problem as follows:
\begin{enumerate}
	\item\label{item1} let $i\in[d]$ be an odd integer, we require $$A[i(B+1)]>A[i(B+1)-s^{(i)}_1]>A[i(B+1)-s^{(i)}_2]>\cdots>A[i(B+1)-s^{(i)}_u],$$ and $$A[i(B+1)]>A[i(B+1)+s^{(i+1)}_1]>A[i(B+1)+s^{(i+1)}_2]>\cdots>A[i(B+1)+s^{(i+1)}_u],$$ they are all smaller than all other entries in $A[(i-1)(B+1)+1,\ldots,(i+1)(B+1)-1]$;
	\item\label{item2} for odd $i$, the order of all elements in 
	\[
	    \{A[i(B+1)-s^{(i)}_j]:j\in[u]\}\cup \{A[i(B+1)+s^{(i+1)}_j]:j\in[u]\}
	\]
	is determined by the integer $k_{(i+1)/2}$ (note that given the above requirement, there are exactly $\binom{2u}{u}$ different orderings);
	\item\label{item3} the Cartesian tree of the subarray $A[i(B+1)-s^{(i)}_{j+1}+1,\ldots,i(B+1)-s^{(i)}_j-1]$ is determined by the integer $z^{(i)}_j$, and the Cartesian tree of the subarray $A[i(B+1)+s^{(i+1)}_{j}+1,\ldots,i(B+1)+s^{(i+1)}_{j+1}-1]$ is determined by the integer $z^{(i+1)}_j$, for $j=0,\ldots,u$;
	\item construct an arbitrary $A$ that satisfies the above constraints.
\end{enumerate}
Then we construct the RMQ data structure for $A$.

\paragraph{Space usage.} By assumption, the RMQ data structure uses $2n+r$ bits of space.
On the other hand, by Lemma~\ref{lem_input}, we have
\begin{align*}
	H_{d,u,B,Z}&=d(2B-u-\Theta(\log B))+\log Z \\
	&=2r(2n/(2r)-u-\Theta(\log B)+r\cdot (2u-\Theta(\log B)) \\
	&=2n-\Theta(d\log B).
\end{align*}
Thus, the space usage of the \predz{} data structure is $H_{d,u,B,Z}+O(d\log B)$ bits, as $r=O(d)$.

\paragraph{Query algorithm.}
To answer the query $\pred(i, x)$, if $i$ is odd, we make the query $\rmq(i(B+1)-x, i(B+1))$ to the RMQ data structure, and return $i(B+1)-\rmq(i(B+1)-x, i(B+1))$.
Similarly, if $i$ is even, we return $\rmq((i-1)(B+1), (i-1)(B+1)+x)-(i-1)(B+1)$.
The query time is $O(t)$.

To see why \texttt{query-z} can be answered, it suffices to recover each of $k_1,\ldots,k_r$ and $\{z^{(i)}_j\}$.
We first ask all $\pred(i, x)$ queries to recover the sets $S_1,\ldots,S_d$.
Then each $k_i$ can be recovered by asking all queries of form
\[
	\rmq((2i-1)(B+1)-s^{(2i-1)}_{j_1}, (2i-1)(B+1)+s^{(2i)}_{j_2}),
\]
for $j_1,j_2\in [u]$.
Since no entries in between can be the minimum, these queries directly compare $A[(2i-1)(B+1)-s^{(2i-1)}_{j_1}]$ with $A[(2i-1)(B+1)+s^{(2i)}_{j_2}]$ for all $j_1$ and $j_2$.
By Item~\ref{item2} above, their answers determine $k_i$.
By Item~\ref{item3}, for odd $i$, $z^{(i)}_j$ can be recovered by asking all \rmq{} queries in the subarray
\[
	A[i(B+1)-s^{(i)}_{j+1}+1,\ldots,i(B+1)-s^{(i)}_j-1],
\]
and for even $i$, $z^{(i)}_j$ can be recovered by asking all queries in
\[
	A[(i-1)(B+1)+s^{(i)}_{j}+1,\ldots,(i-1)(B+1)+s^{(i)}_{j+1}-1].
\]

\bigskip

Finally, by Lemma~\ref{lem_predz}, we have
\[
	(wt\log B)^{O(t^2\log^2 t)}\geq B=\Omega(n/r).
\]
Since $w\geq \Omega(\log n)$, we have
\[
	r\geq \frac{n}{w^{O(t^2\log^2 t)}}.
\]
This proves the theorem.

\end{proof}

\subsection{\predz{} lower bound}

In this subsection, we prove Lemma~\ref{lem_predz}.
The proof strategy is based on a variant of round elimination of~\cite{PV10}.
We will focus on the predecessor search queries on all \emph{input data points}, which turn out to be the hard queries.
That is, let $\cQ:=\{\pred(i, s^{(i)}_j):i\in [d],j\in[u]\}$.
We will prove a lower bound on the \emph{expected average} query time of all queries in $\cQ$, when the inputs $\{S_1,\ldots,S_u\}$ and $z$ are \emph{uniformly random}.
Note that the set of queries $\cQ$ is also random, as it depends on the input.

During the round elimination, we will work with data structures with \emph{published bits}.
More specifically, at the beginning of each round, we are given a data structure for \predz{} using optimal $H_{d,u,B,Z}$ bits of memory, with extra $p$ published bits.
The query algorithm may access these published bits for free, as well as probing the regular memory cells with standard cost of one per probe.
In each round, we will create a new data structure with more published bits, but faster expected average query time of $\cQ$ (probing fewer regular memory cells).

To this end, for any fixed data structure, we denote by $\probe(q)$, the set of memory cells probed in order to answer query $q$.
Similarly for any set of queries $Q$, $\probe(Q):=\bigcup_{q\in Q}\probe(q)$.
We have the following lemma.

\begin{lemma}\label{lem_int}
	Given a \predz{} data structure with $p$ published bits for $p\geq d$ and $p<du\cdot \log^{-4} B$ and worst-case $\pred$ query time $t$, there exists a set $Q_{\pub}$ of $p\log^4 B$ $\pred(\cdot,\cdot)$ queries, possibly random and depending on the input, such that
	\[
		\E\left[\frac{1}{|\cQ|}\sum_{q\in \cQ} |\probe(q)\cap \probe(Q_{\pub})|\right]\geq \Omega\left(\frac{1}{t\log^2 t}\right),
	\]
	where the expectation is taken over uniformly random input data $\{S_1,\ldots,S_d\}$, $z$ and the choice of $Q_{\pub}$.
\end{lemma}

The proof of lemma is deferred to the next subsection.
Now, let us prove that it implies Lemma~\ref{lem_predz}.

\begin{proof}[Proof of Lemma~\ref{lem_predz}]
Suppose there is a \predz{} data structure using $O(d\log B)$ bits of redundancy and query time $t$.
To initialize the round elimination argument, we simply publish the last $O(d\log B)$ bits of the data structure, and keep the first $H_{d,u,B,Z}$ bits in memory.
Thus, we obtain a data structure with $p=O(d\log B)$ publish bits, worst-case query time $t$ and expected average query time of $\cQ$ also at most $t$.

In each round, we begin with a data structure $D$ with $p$ published bits, and will modify it to a new data structure with a lower query time for $\cQ$.
To this end, we first construct $D$ given the input data.
Then, we apply Lemma~\ref{lem_int} to find the set $Q_{\pub}$, possibly depending on the inputs, and further publish all cells in $\probe(Q_{\pub})$.
That is, we append the addresses and contents of all cells in $\probe(Q_{\pub})$ to the published bits.
Thereafter, when the query algorithm wants to probe a cell, it first checks if the cell is already published, by reading the published bits.
If it is, this probe can be avoided, as the published bits already have the contents.

By the guarantee of Lemma~\ref{lem_int}, for an average query in $\cQ$, in expectation $\Omega(1/(t\log^2 t))$ probes are avoided, i.e., the expected average query time of $\cQ$ is reduced by $\Omega(1/(t\log^2 t))$.
On the other hand, publishing $\probe(Q_{\pub})$ takes $O(p\cdot wt\log^4 B)$ bits.
If we were able to execute this argument for more than $O(t^2\log^2 t)$ rounds, the query time of $\cQ$ would become a negative number, which is a contradiction.
Hence, $p$ must have exceeded $du\log^{-4} B$ before it happens, so that the premises of Lemma~\ref{lem_int} become unsatisfied.
Therefore, we must have
\[
	O(d\log B)\cdot (wt\log^4 B)^{O(t^2\log^2 t)}\geq du\log^{-4} B,
\]
which by the fact that $u=\Theta(\sqrt{B})$, simplifies to 
\[
	(wt\log B)^{O(t^2\log^2 t)} \geq B.
\]

\end{proof}

\ifx\mainfile\undefined
\end{document}
\fi

\ifx\lbfile\undefined
\documentclass[11pt]{article}

\begin{document}
\fi

\subsection{Selecting queries $Q_{\pub}$}\label{sec_qpub}

The main argument of the proof lies in finding such set of queries $Q_{\pub}$.
The intuition is that we want to select queries that reveal a lot of information (much more than $p$ bits) about the inputs.
Then since all queries in $\cQ$ determine the whole input, if the sets of cells $Q_{\pub}$ and $\cQ$ probe are always very different and barely intersect, it would mean that more than $p$ bits of information must have been stored in two different locations in the data structure.
However, the data structure only uses $p$ extra bits beyond the information theoretical minimum (since we assume the input is uniformly random), we derive a contradiction.
The formal argument is more complex, which we elaborate below.

\bigskip

It turns out that $Q_{\pub}$ can be selected by picking $p/d$ evenly spaced (over the $u$ elements) elements from each set $S_i$, then selecting $\poly\log n$ evenly spaced (over $[B]$) queries between the picked elements.
More specifically, recall that $S_i=\{s^{(i)}_1,\ldots,s^{(i)}_u\}$, let $m=ud/p$ be the gap between the picked elements, and
\[
	S^{(i)}_{\pt}:=\left\{s^{(i)}_{m},s^{(i)}_{2m},s^{(i)}_{3m},\ldots,s^{(i)}_u\right\}
\]
be $u/m=p/d$ evenly spaced elements in $S_i$.
$S^{(i)}_{\pt}$ partitions the universe $[B]$ into $p/d$ blocks, each block has $m$ elements from $S_i$, but the blocks may have different lengths.
In the following, we will only focus on the blocks whose length is approximately $m^2$.
Let $\cS_{\good}$ be the disjoint union of all such blocks $[x, y]$ from all $S_i$:
\[
	\cS_{\good}:=\left\{(i,[x, y]):i\in[d],\frac{1}{2}m^2\leq y-x\leq 2m^2,\exists l\in [p/d],\textrm{s.t.},x=s^{(i)}_{lm}+1,y=s^{(i)}_{(l+1)m}-1\right\}.
\]
Finally, we pick approximately $\log^4 B$ evenly spaced points in each block in $\cS_{\good}$, and the $\pred{}$ queries on them will form the set $Q_{\pub}$ (note that this is possible when $p<du\log^{-4} B$).
Formally, let $L=m^2\cdot \log^{-4} B$, and $\Delta\in [L]$ be uniformly random, we define
\[
	Q_{\pub}:=\bigcup_{(i,[x, y])\in \cS_{\good}}\left\{\pred(i,x+j\cdot L+\Delta):j\geq 1,x+j\cdot L+\Delta<y\right\}.
\]
The size of $Q_{\pub}$ is at most $O(p\log^4 B)$, since $|\cS_{\good}|\leq p$ and $y-x\leq 2L\log^4 B$. 
Note that $\cS_{\good}$, together with $\Delta$, determines $Q_{\pub}$.
From now on, it is helpful to view each block $(i,[x, y])\in \cS_{\good}$ as an independent universe, since $x-1$ is an input data point, any $\pred$ queries asked in this range must also have its answer in it (or equal to $x-1$).

As we argued above, we will need to show that the answers to $Q_{\pub}$ reveal a lot of information about the input.
It turns out that the most ``useful'' information they reveal is whether there is any input data point between two adjacent queries.
Define the indicator variable
\[
	E^{(i,[x, y])}_j:=\mathbf{1}_{\pred\left(i, x+jL+\Delta\right)\neq \pred\left(i, x+(j+1)L+\Delta\right)},
\]
indicating if there is an input point between them.
We can show that for each block in $\cS_{\good}$, the joint entropy of $\{E^{(i,[x, y])}_j\}_{j}$ is large.
\begin{lemma}\label{lem_entropy_lb}
	Let the input $(S_1,\ldots,S_d)$ and $z$ be uniformly random, conditioned on $(i,[x, y])\in\cS_{\good}$, the joint entropy is large for a random offset $\Delta$,
	\[
		H\left(E^{(i,[x, y])}_1,\ldots,E^{(i,[x, y])}_{(y-x)/L}\mid (i,[x, y])\in\cS_{\good},\Delta\right)\geq \Omega(\log^2 B).
	\]
\end{lemma}
We can also show that most blocks are good.
\begin{lemma}\label{lem_good_blocks}
	Let the inputs $(S_1,\ldots,S_d)$ and $z$ be uniformly random.
	For every $i\in[d],l\in[p/3d,2p/3d]$, we have
	\[
		\Pr[(i,[x, y])\in\cS_{\good}]\geq \Omega(1),
	\]
	where $x=s^{(i)}_{lm}+1$ and $y=s^{(i)}_{(l+1)m}-1$.
\end{lemma}

On the other hand, if we are given a large subset of the $m$ input points in a block, then $\{E^{(i,[x, y])}_j\}_{j}$ can be described succinctly.
\begin{lemma}\label{lem_enc_len}
	Let the input $(S_1,\ldots,S_d)$ and $z$, as well as $\Delta\in[L]$, be uniformly random, and condition on $(i,[x, y])\in\cS_{\good}$.
	Let $S_{i, [x,y]}'\subseteq S_i\cap [x, y]$ be an (arbitrarily) jointly distributed subset,
	there is a prefix-free binary string $\com_{\Delta,S_i}(S_{i, [x,y]}')$, such that $\com_{\Delta,S_i}(S_{i, [x,y]}')$ and $S_{i, [x,y]}'$ together determine $\{E^{(i,[x, y])}_j\}_j$.
	Moreover, we have the following bound on the length of $\com_{\Delta,S_i}(S_{i, [x,y]}')$:
	\[
		\E
		\left[
		  \left|\com_{\Delta,S_i}(S_{i, [x,y]}')\right|
		  \mid (i,[x, y])\in\cS_{\good}
		\right]\leq 
		O(\sqrt{\epsilon} \log^2 B\log(1/\epsilon)+\log B\log\log B),
	\]
	where $\epsilon:=1-\frac{1}{m}\E\left[\left|S_{i, [x,y]}'\right|\mid (i,[x, y])\in\cS_{\good}\right]$.

\end{lemma}

The proofs of the above three lemmas highly rely on the input distribution, especially on the marginal of each $S_i$.
To focus on the main storyline, we defer them to Section~\ref{sec_dist}.
For now, the only property we use about the uniform input distribution is that conditioned on $\cS_{\good}$, the data points in different blocks are mutually independent, which can be seen easily from definition.
\begin{claim}\label{cl_ind}
  Let $S_\pt:=\{S_\pt^{(i)}\}_i$ be the partitions.
  $S_\pt$ determines $\cS_\good$, and given $S_\pt$, all $p$ blocks $ \left\{(i,[s^{(i)}_{lm}+1,s^{(i)}_{(l+1)m}-1]):i\in[d],l\in[p/d]\right\} $ are mutual independent.
\end{claim}

Now, we are ready to prove Lemma~\ref{lem_int}.

\begin{proof}[Proof of Lemma~\ref{lem_int}]
The proof uses an encoding argument.
We will show that if the intersection size in the lemma statement is too small, then there is a very efficient way to encode the entire input data, using bits less than its entropy, which yields a contradiction.

To this end, a useful notion (following~\cite{PV10}) is the \emph{footprint} of a query.
Given a query $q$, its footprint $\foot(q)$ is a binary string obtained by concatenating the ($w$-bit) content of the first cell the query algorithm probes when answering $q$, the content of the second cell, third cell, and so on.
This is a $w\cdot |\probe(q)|$-bit binary string, encoding all contents in $\probe(q)$.
Note that it is not necessary to write down the addresses, since by simulating the query algorithm, one automatically knows what is the next cell to probe, given $q$.
Therefore, $\foot(q)$, together with $q$, determines the answer to $q$.

Similarly, given a set of queries $Q$, its footprint $\foot(Q)$ is obtained by concatenating the contents of the first probes of all queries in $Q$ (in the lexicographical order of corresponding query), then the contents of the second probes of all queries, and so on.
If a cell is already encoded, either in the same $l$-th probe by a lexicographically smaller query or in an earlier probe, its content will be skipped.
Likewise, $\foot(Q)$, together with $Q$, also determines the answers to all queries in $Q$.

Finally, for $l\in [t]$, let $\foot_{<l}(Q)$ be the prefix of $\foot(Q)$ that only encodes the cells in the first $(l-1)$ probes of $Q$.
Let $\foot_l(Q)$ be the substring of $\foot(Q)$ that encodes the $l$-th probes.
Note that given $\foot_l(Q)$ alone, it might not be clear which cells it is encoding, since the location of the $l$-th probe may depend on the contents of the previous probes.
It may also not encode the $l$-th probe of all queries in $Q$, as some of them may have appeared in $\foot_{<l}(Q)$, which are skipped by the definition of $\foot(Q)$.

\bigskip
Next, let $\cE$ be the set of all indicator random variables $\{E^{(i,[x, y])}_j:(i,[x, y])\in \cS_{\good},j\in[(y-x)/L]\}$.
By Lemma~\ref{lem_entropy_lb}, Lemma~\ref{lem_good_blocks} and Claim~\ref{cl_ind}, we have
\[
	H(\cE \mid S_{\pt},\Delta)\geq \Omega(p\log^2 B).
\]
The answers to $Q_{\pub}$ determine $\cE$, hence,
\[
	I(\foot(Q_{\pub}); \cE\mid S_{\pt},\Delta)\geq \Omega(p\log^2 B).
\]
By the chain rule of mutual information, we have
\begin{equation}\label{eqn_cr}
	\sum_{l=1}^t I(\foot_l(Q_{\pub});\cE\mid S_{\pt},\foot_{<l}(Q_{\pub}),\Delta)\geq \Omega(p\log^2 B).
\end{equation}
Note that here it is clear which contents $\foot_l(Q_{\pub})$ is encoding, since we have conditioned on $\foot_{<l}(Q_{\pub})$, $\cS_{\good}$ and $\Delta$.
Equivalently, it might be helpful to think $\foot_l(Q_{\pub})$ as a set of cells with \emph{both} contents and addresses encoded, however, the entropy of the addresses is zero conditioned on $\foot_{<l}(Q_{\pub})$.

\bigskip

Now, let us assume
\begin{equation}\label{eqn_int}
	\E_{\Delta,S_1,\ldots,S_d,z}\left[\frac{1}{|\cQ|}\sum_{q\in \cQ} |\probe(q)\cap \probe(Q_{\pub})|\right]=\delta.
\end{equation}
The goal is to lower bound $\delta$.

\paragraph{Encoding.}
Consider the following encoding scheme that encodes the data structure:
\begin{enumerate}
	\item
		write down the $p$ published bits;
	\item
		write down $S_{\pt}$;
	\item
		sample a uniformly random $\Delta\in[L]$ and $l\in[t]$, write down $\Delta$ and $l$;
	\item
		write down $\foot_{<l}(Q_{\pub})$;
	\item\label{step5}
	  	write down the contents of all cells that are not in $\foot_{<l}(Q_{\pub})$ or $\foot_l(Q_{\pub})$, in the increasing order of their addresses;
	\item
		for every block $(i, [x, y])\in\cS_{\good}$, let $S'_{i,[x, y]}$ be the subset of $S_i\cap [x, y]$ such that $s\in S'_{i,[x, y]}$ if and only if $\pred(i, s)$ does not probe any cell in $\foot_l(Q_{\pub})$;
	\item
		for each $(i, [x, y])\in\cS_{\good}$, apply Lemma~\ref{lem_enc_len} and write down $\com_{\Delta,S_i}(S'_{i,[x, y]})$;
	\item
		encode $\foot_l(Q_{\pub})$ conditioned on $\cE, S_{\pt}$ and $\foot_{<l}(Q_{\pub})$ using the optimal expected $$H(\foot_l(Q_{\pub})\mid \cE, S_{\pt}, \foot_{<l}(Q_{\pub}),\Delta)+O(1)$$ bits, and write down the encoding.
  \end{enumerate}

\paragraph{Decoding.}
	Now, we show that one can recover the entire data structure from the above encoding, hence, all the inputs.

	Given the above encoding, we first read the $p$ published bits, $S_{\pt}$, $\Delta$ and $l$.
	From them, we know the set $Q_{\pub}$.
	By simulating the query algorithm for $(l-1)$ steps on $Q_{\pub}$, we read $\foot_{<l}(Q_{\pub})$ and recover all their contents.
	At the same time, we know the \emph{addresses} of all cells in $\foot_l(Q_{\pub})$ (but not their contents).
	Next, we read the contents of all cells not in $\foot_{<l}(Q_{\pub})$ or $\foot_l(Q_{\pub})$ from Step~\ref{step5}.

	So far, we have recovered contents of all cells not in $\foot_l(Q_{\pub})$, and we know their addresses.
	We go over all possible $\pred(\cdot,\cdot)$ queries, and simulate the query algorithm on them.
	This identifies all queries that can be answered without probing cells in $\foot_l(Q_{\pub})$.
	In particular, we recover the sets $S'_{i,[x, y]}$ for every $(i,[x, y])\in\cS_{\good}$.

	Next, we read $\com_{\Delta,S_i}(S'_{i,[x, y]})$.
	By Lemma~\ref{lem_enc_len}, together with $S'_{i,[x, y]}$, we recover all $E^{(i,[x, y])}_j$ for all $(i,[x, y])\in \cS_{\good}$, i.e., we recover $\cE$.
	Finally, we read the encoding of $\foot_l(Q_{\pub})$ conditioned on $\cE$, $\cS_{\good}$, and $\foot_{<l}(Q_{\pub})$.
	This reconstructs the data structure, and thus, by making all queries to it, we recover the entire input data.
	Therefore, the encoding must use at least $H_{d,u,B,Z}$ bits in expectation.

\paragraph{Analysis.}
	Now let us analyze how many bits the above encoding scheme takes in expectation:
	\begin{enumerate}
		\item
			published bits take $p$ bits;
		\item
			$S_{\pt}$ has $p$ blocks, each taking $O(\log B)$ bits;
		\item
			$\Delta$ and $l$ take $O(\log B+\log t)$ bits;
		\item
			$\foot_{<l}(Q_{\pub})$ take $|\foot_{<l}(Q_{\pub})|$ bits;
		\item
			all other cells take $H_{d,u,B,Z}-|\foot_{<l}(Q_{\pub})|-|\foot_{l}(Q_{\pub})|$ bits;
		\item
			by Equation~\eqref{eqn_int} and Lemma~\ref{lem_good_blocks}, for a constant fraction of the queries $q\in\cQ$, the probability that it is in a good block is at least a constant, and we choose $l\in[t]$ uniformly at random, in expectation at most $O(\delta/t)$ fraction of the hard queries in each good $S_i\cap [x, y]$ probe cells in $\foot_l(Q_{\pub})$, i.e.,
			\[
			  \E_{S_\pt,\Delta,l,\{S_i\}|_{S_\pt}}\left[\frac{1}{|\cS_{\good}|}\cdot \sum_{(i,[x, y])\in\cS_{\good}}|S'_{i,[x, y]}|\right]\geq (1-O(\delta/t))m;
			\]
		\item
		  for each $(i,[x, y])\in\cS_{\good}$, by Lemma~\ref{lem_enc_len}, we have
			\[
			  |\com_{\Delta,S_i}(S'_{i,[x, y]})|\leq O((\sqrt{\epsilon_{s_\pt,(i,[x,y])}}\log 1/\epsilon_{s_\pt,(i,[x,y])})\log^2 B+\log B\log\log B),
			\]
			where $\epsilon_{s_\pt,(i,[x,y])}:=1-\E[|S'_{i,[x, y]}|\mid S_\pt=s_\pt]/m$.
			Let $\epsilon_{s_\pt}:=\E_{(i,[x,y])\in \cS_\good}[\epsilon_{s_\pt,(i,[x,y])}], \epsilon:=\E[\epsilon_{s_\pt}]=O(\delta/t)$.
			Writing down all $\com_{\Delta,S_i}(S'_{i,[x, y]})$ takes
			\begin{align*}
			  &\ \sum_{s_\pt}\Pr[S_\pt=s_\pt]\sum_{(i,[x,y])\in \cS_\good}\E[|\com_{\Delta,S_i}(S'_{i,[x, y]})|\mid S_\pt=s_\pt]\\
			  =&\ \E_{s_\pt}\left[|\cS_\good|\E_{(i,[x,y])\in \cS_\good}[|\com_{\Delta,S_i}(S'_{i,[x, y]})|\mid S_\pt=s_\pt]\right]\\
			  =&\ \E_{s_\pt}\left[|\cS_\good| \E_{(i,[x,y])\in \cS_\good}[O((\sqrt{\epsilon_{s_\pt,(i,[x,y])}}\log 1/\epsilon_{s_\pt,(i,[x,y])})\log^2 B+\log B\log\log B)\mid S_\pt=s_\pt]\right]\\
			  \le&\ p\log^2B\cdot\E_{s_\pt}\left[O(\sqrt{\epsilon_{s_\pt}}\log 1/\epsilon_{s_\pt})\right]+O(p\log B\log\log B)\\
			  \le&\ p\log^2B\cdot O(\sqrt{\epsilon}\log 1/\epsilon)+O(p\log B\log\log B)\\
			  =&\ O(p\log^2B\cdot\sqrt{\delta/t}\log(t/\delta)+p\log B\log\log B)
			\end{align*}
			bits in expectation, due to the concavity of $\sqrt{x}\log(1/x)$ for $0<x<1$;
		\item
			by Equation~\eqref{eqn_cr}, encoding $\foot_{l}(Q_{\pub})$ takes
			\[
				|\foot_{l}(Q_{\pub})|-\Omega((p\log^2 B)/t)
			\]
			bits.
	\end{enumerate}
	Summing up the cost of each step, the encoding uses in total
	\begin{align*}
		&H_{d,u,B,Z}+p+O(p\log B)+O(\log B+\log t)\\
		&+O(p(\sqrt{\delta/t}\log (t/\delta))\log^2 B)+O(p\log B\log\log B)-\Omega((p\log^2 B)/t) \\
		\leq&\, H_{d,u,B,Z}+O(p(\sqrt{\delta/t}\log (t/\delta))\log^2 B)+O(p\log B\log\log B)-\Omega((p\log^2 B)/t)
	\end{align*}
	bits.
	Since $t\ll \log B/\log\log B$ (otherwise the lower bound already holds), and thus, $p\log B\log\log B\ll (p\log^2 B)/t$, we must have
	\[
		p(\sqrt{\delta/t}\log (t/\delta))\log^2 B\geq \Omega((p\log^2 B)/t),
	\]
	which simplifies to
	\[
		\delta \geq \Omega(1/(t\log^2 t)).
	\]
	This proves the lemma.
\end{proof}


\ifx\lbfile\undefined
\end{document}
\fi



\section{Analyzing the Input Distribution}\label{sec_dist}

In this section, we analyze the input distribution and prove the three lemmas from the previous section.
Recall that the input to the \predz{} problem is $(S_1,\ldots,S_d)$ and $z$ such that
\[
    1\leq z\leq Z\cdot \prod_{i=1}^d\prod_{j=0}^u C_{s^{(i)}_{j+1}-s^{(i)}_j-1}.
\]
Therefore, if we sample a uniformly random input, the marginal distribution of each $S_i$ is
\begin{equation}\label{eqn_prob_S}
	\Pr[S_i=\{s_1,\ldots,s_u\}]=\frac{1}{M(B,u)}\prod_{j=0}^u C_{s_{j+1}-s_j-1},
\end{equation}
where $0=s_0<s_1<\cdots<s_u<s_{u+1}=B+1$, and
\[
	M(B,u):=\sum_{s_1,\ldots,s_u}\prod_{j=0}^u C_{s_{j+1}-s_j-1}
\]
is the Catalan's $(u+1)$-fold convolution.
In general, we have the following equation.
\begin{theorem}[Catalan's $m$-fold convolution formula~\cite{catalan1887nombres,regev2012proof}]
For $m\leq U$,
\begin{align}
    \sum_{\substack{i_1+\dots+i_m=U\\i_1,\dots,i_m\ge 1}}C_{i_1-1}\dots C_{i_m-1}=\frac{m}{2U-m}\binom{2U-m}{U}.
\end{align}
\end{theorem}
Thus, we have the following estimation for $M$.
\begin{proposition}\label{prop_M}
	For super constant $B$ and any $u\ll B$, it holds that
	\[
	    M(B,u) \leq \frac{(1+o(1))2u}{\sqrt{\pi/2}(2B-u)^{1.5}}\cdot 2^{2B-u}e^{-u^2/(4B)},
 	\]
 	and when $u^3\ll B^2$,
 	\[
	    M(B,u)\geq \frac{(1-o(1))2u}{\sqrt{\pi/2}(2B-u)^{1.5}}\cdot 2^{2B-u}e^{-u^2/(4B-2u)}.
	\]
	In particular, 
	when $u=O(\sqrt{B})$,
	\[
	    M(B,u)=\Theta(u2^{2B-u}/B^{3/2}).
	\]
\end{proposition}
\begin{proof}
\begin{align*}
    M(B, u)&=\frac{u+1}{2B-u+1}\binom{2B-u+1}{B+1}\\
    &=\frac{u+1}{2B-u+1}\binom{2B-u+1}{B-\lfloor (u-1)/2\rfloor }\frac{(B-u+1)\cdots(B-\lceil (u-1)/2\rceil )}{(B-\lfloor (u-1)/2\rfloor +1)\cdots (B+1)}\\
	&=\frac{(1\pm o(1))u2^{2B-u+1}}{\sqrt{\pi/2}(2B-u)^{1.5}}\cdot\frac{(B-u+1)\cdots(B-\lceil (u-1)/2\rceil)}{(B-\lfloor (u-1)/2\rfloor +1)\cdots (B+1)}.
\end{align*}
The last factor is upper bounded by
\begin{align*}
    \frac{(B-u+1)\cdots(B-\lceil (u-1)/2\rceil)}{(B-\lfloor (u-1)/2\rfloor +1)\cdots (B+1)}&\leq \left(\frac{B-\lceil (u-1)/2\rceil }{B+1}\right)^{\lfloor (u+1)/2\rfloor } \\
    &\leq \left(1-\frac{\lceil (u+1)/2\rceil }{B+1}\right)^{\lfloor (u+1)/2\rfloor } \\
    &\leq (1+o(1))e^{-u^2/(4B)},
\end{align*}
and when $u^3\ll B^2$, it is lower bounded by 
\begin{align*}
    \frac{(B-u+1)\cdots(B-\lceil (u-1)/2\rceil)}{(B-\lfloor (u-1)/2\rfloor +1)\cdots (B+1)}&\geq \left(\frac{B-u}{B-\lfloor (u-1)/2\rfloor}\right)^{\lfloor (u+1)/2\rfloor} \\
    &\geq \left(1-\frac{u}{2B-u-1}\right)^{u/2} \\
    &\geq (1-o(1))e^{-u^2/(4B-2u)}.
\end{align*}

\end{proof}

Now, let us first prove Lemma~\ref{lem_good_blocks}.
\begin{restate}[Lemma~\ref{lem_good_blocks}]
	Let the inputs $(S_1,\ldots,S_d)$ and $z$ be uniformly random.
	For every $i\in[d],l\in[p/3d,2p/3d]$, we have
	\[
		\Pr[(i,[x, y])\in\cS_{\good}]\geq \Omega(1),
	\]
	where $x=s^{(i)}_{lm}+1$ and $y=s^{(i)}_{(l+1)m}-1$.
\end{restate}
\begin{proof}
    Let us first fix $i$, and omit the superscript $(i)$ in the following for convenience of notations.
    Let $c=lm$, it suffices to prove that $\Pr[s_{c+m}-s_c\in [m^2/2, 2m^2]]\geq \Omega(1)$.
  Note that $u^2=\Theta(B)$, we have
  \begin{align*}
	&\Pr[s_{c+m}=y+1,s_c=x-1]\\
	=&\, M(x-2,c-1)\cdot M(y-x+1,m-1)\cdot M(B-y,u-c-m)/M(B,u)\\
	=&\, \Omega\left(\frac{cm(u-c-m)B^{1.5}}{x^{1.5}(y-x)^{1.5}(B-y)^{1.5}u}\cdot e^{-\frac{c^2}{4x-2c}-\frac{m^2}{4(y-x)-2m}-\frac{(u-c-m)^2}{4(B-y)-2(u-c-m)}+\frac{u^2}{4B}}\right).
  \end{align*}
  Since $m\ll u$ and $l\in [p/3d,2p/3d]$ (i.e., $c\in[u/3,2u/3]$), we have
  \begin{align*}
	&\Pr[s_{c+m}-s_c\in[m^2/2,2m^2]]\\
	\geq&\,\sum_{x=B/3}^{2B/3}\sum_{y-x=m^2/2}^{2m^2}\Pr[s_{c+m}=y,s_c=x]\\
	=&\, \sum_{x=B/3}^{2B/3}\sum_{y-x=m^2/2}^{2m^2}\Omega\left(\frac{m}{B(y-x)^{1.5}}e^{-\frac{m^2}{4(y-x)-2m}}\right)\\
	\geq &\, \Omega\left(B\cdot m^2\cdot \frac{m}{B(m^2)^{1.5}}\right) \geq \Omega(1).
  \end{align*}
\end{proof}


In the rest of this section, we will prove the two remaining lemmas.
Both lemmas consider one specific good block $(i, [x, y])$.
For simplicity of notations, we will omit all the superscripts $(i, [x, y])$.

We denote the length of the block by $U:=y-x+1$ and $S=S_i\cap [x, y]$ be the set of input points contained in it.
Thus, we have $U=\Theta(m^2)$.
It is easy to verify that the marginal distribution of $S$ is similar to Equation~\eqref{eqn_prob_S}:
\[
	\Pr[S=\{s_1,\ldots,s_m\}]=\frac{1}{M(U,m)}\prod_{j=0}^m C_{s_{j+1}-s_j-1}.
\]

Let $k=\log^4 B$.
We have $L=\lfloor m^2/k\rfloor$.
Let $K=\lfloor U/L\rfloor$ be approximately the number of predecessor search queries we put in $Q_{\pub}$ in $[x, y]$.
We sample a random integer $\Delta\in[1, L]$.
Recall that within this block we put in $Q_{\pub}$ all predecessor search queries at points of form $\Delta+j\cdot L$.
That is, we divide $U$ into roughly $K$ intervals of size $L$, with a uniformly random offset.
Finally, recall that $E_j$ is the indicator random variable indicating $[\Delta+j\cdot L,\Delta+(j+1)\cdot L)$ contains at least one point in $S$.



\bigskip
We will involve the following inequalities as a toolkit to finish our analysis.
The proofs are deferred to the appendix.
\begin{fact}\label{fact_sum1}
	For $A,T>0$, $\sum_{h=1}^T e^{-A/h}\cdot h^{-1.5}\leq O(e^{-A/T}(1/\sqrt{T}+1/\sqrt{A}))$.
\end{fact}
\begin{fact}\label{fact_sum2}
	For $A,T>0$, $\sum_{h=1}^T e^{-A/(T-h)}\cdot h^{-0.5}\leq O(e^{-A/T}\sqrt{T})$.
\end{fact}
\begin{fact}\label{fact_sum3}
    For $A_1,A_2,T>0$, 
	\begin{align*}
	    &\sum_{h=1}^T e^{-A_1/h-A_2/(T-h)}\cdot h^{-1.5} (T-h)^{-1.5} \\
	    \leq&\, O\left(T^{-1.5}e^{-(A_1+A_2)/2T}\cdot\left(1+\sqrt{(A_1+A_2)/T}+\sqrt{\frac{A_1A_2}{T(A_1+A_2)}}\right)\cdot\left(1/\sqrt{A_1}+1/\sqrt{A_2}\right)\right).
	\end{align*}
\end{fact}
\begin{fact}\label{fact_conv}
	For positive $2A_1\leq A_2$ and $B_1\geq B_2$, we have
	$
		\frac{A_1^2}{B_1}+\frac{7}{8}\cdot \frac{A_2^2}{B_2}\geq \frac{(A_1+A_2)^2}{B_1+B_2}.
		$
\end{fact}

The following lemma asserts that not too many non-empty intervals have few elements.
\begin{lemma}\label{lem_small_interval}
  For $l\leq m/\sqrt{k}=O(\sqrt{L})$, the expected number of intervals that have between $l/2$ and $l$ elements is at most $O(lk/m+1)$.
\end{lemma}
\begin{proof}
	Instead of upper-bounding the expected number of such intervals directly, we are going to count how many elements can be the \emph{first} element in such a interval in expectation.
	By linearity of expectation, it suffices to estimate for each element $s_c$, what is the probability that it becomes such an element.

	Fix an integer $c\in[2, m-3l]$.
	$s_c$ can be the first element in a interval that have at most $l$ elements, only if
	\begin{itemize}
		\item $s_{c-1}$ is in an earlier interval, and
		\item $s_{c+l}$ is in a later interval.
	\end{itemize}

	To bound the probability, let us first estimate the probability that $s_{c+l}=z$ conditioned on $s_c=y$ for some $y\leq U-3L$ and $z>y$.
	By Proposition~\ref{prop_M}, we have
	\begin{align*}
		\Pr[s_{c+l}=z\mid s_c=y]&=\frac{M(z-y-1,l-1)\cdot M(U-z,m-c-l)}{M(U-y,m-c)} \\
		&=O\left(\frac{l(m-c-l)2^{2(U-y)-(m-c)}\cdot \exp(-\frac{l^2}{4(z-y)}-\frac{(m-c-l)^2}{4(U-z)})}{M(U-y,m-c)\cdot (z-y)^{1.5}(U-z)^{1.5}}\right).
	\end{align*}
	By Cauchy-Schwarz, $\frac{l^2}{4(z-y)}+\frac{(m-c-l)^2}{4(U-z)}\geq \frac{(m-c)^2}{4(U-y)}$.
	Thus, when $z\leq\frac{U+y}{2}$ (which implies $1/(U-z)\le 2/(U-y)$), the probability is at most
	\[
		\Pr[s_{c+l}=z\mid s_c=y]\leq O\left(\frac{l(m-c-l)2^{2(U-y)-(m-c)}}{M(U-y,m-c)}\cdot \frac{\exp(-\frac{(m-c)^2}{4(U-y)})}{(z-y)^{1.5}(U-y)^{1.5}}\right).
	\]
	When $z>\frac{U+y}{2}$ (which implies $1/(U-y)\ge 1/(2z-2y)$), since $2l<m-c-l$ and $z-y>U-z$, we have
	\begin{align*}
		\Pr[s_{c+l}=z\mid s_c=y]
		&=O\left(\frac{l(m-c-l)2^{2(U-y)-(m-c)}\cdot \exp(-\frac{l^2}{4(z-y)}-\frac{(m-c-l)^2}{4(U-z)})}{M(U-y,m-c)\cdot (z-y)^{1.5}(U-z)^{1.5}}\right).
	\end{align*}
	and by Fact~\ref{fact_conv},
	\begin{align*}
		\frac{l^2}{4(z-y)}+\frac{7}{8}\cdot \frac{(m-c-l)^2}{4(U-z)}\geq \frac{(m-c)^2}{4(U-y)}.
	\end{align*}
	The probability is also at most
	\[
		\Pr[s_{c+l}=z\mid s_c=y]\leq O\left(\frac{l(m-c-l)2^{2(U-y)-(m-c)}}{M(U-y,m-c)}\cdot\frac{e^{-\frac{(m-c)^2}{4(U-y)}-\frac{(m-c-l)^2}{32(U-z)}}}{(U-y)^{1.5}(U-z)^{1.5}}\right).
	\]
	Thus, for $U-y\geq 3L$ and $y<x\leq y+L$, we have
	\begin{align}
		\Pr[s_{c+l}\geq x\mid s_c=y]&=\sum_{z=x}^{U+y-x}\Pr[s_{c+l}=z\mid s_c=y]+\sum_{z=U+y-x+1}^{U}\Pr[s_{c+l}=z\mid s_c=y] \nonumber\\
		&\leq O\left(\frac{l(m-c-l)2^{2(U-y)-(m-c)}}{M(U-y,m-c)}\cdot \frac{e^{-\frac{(m-c)^2}{4(U-y)}}}{(U-y)^{1.5}}\cdot \left(\sum_{z=x}^{U+y-x}\frac{1}{(z-y)^{1.5}}+\sum_{z=U+y-x}^{U}\frac{e^{-\frac{(m-c-l)^2}{32(U-z)}}}{(U-z)^{1.5}}\right)\right)\nonumber\\
		&\leq O\left(\frac{l(m-c-l)2^{2(U-y)-(m-c)}}{M(U-y,m-c)}\cdot \frac{e^{-\frac{(m-c)^2}{4(U-y)}}}{(U-y)^{1.5}}\cdot \left(\frac{1}{\sqrt{x-y}}+\sum_{U-z=1}^{x-y}\frac{e^{-\frac{(m-c)^2}{32(U-z)}}}{(U-z)^{1.5}}\right)\right) \nonumber\\
		\intertext{which by Fact~\ref{fact_sum1}, is at most}
		&\leq O\left(\frac{l(m-c-l)2^{2(U-y)-(m-c)}}{M(U-y,m-c)}\cdot \frac{e^{-\frac{(m-c)^2}{4(U-y)}}}{(U-y)^{1.5}}\cdot \left(\frac{1}{\sqrt{x-y}}+\frac{e^{-\frac{(m-c)^2}{32(x-y)}}}{m-c}\right)\right).\label{eqn_bound_cl}
	\end{align}

	Similarly, for $y\geq 2L$ and $x>y-L$ (which implies $x>y/2$), we have
	\begin{align}
		\Pr[s_{c-1}\leq x\mid s_c=y]&=\sum_{z=1}^{x}\frac{C_{y-z-1}\cdot M(z-1,c-2)}{M(y-1,c-1)} \nonumber \\
		&\leq O\left(\frac{1}{M(y-1,c-1)}\sum_{z=1}^x\frac{c2^{2y-c}e^{-c^2/4z}}{(y-z)^{1.5}z^{1.5}}\right) \nonumber \\
		&\leq O\left(\frac{c2^{2y-c}}{M(y-1,c-1)}\left(\sum_{z=y-x}^x\frac{e^{-c^2/4z}}{(y-z)^{1.5}z^{1.5}}+\sum_{z=1}^{y-x}\frac{e^{-c^2/4z}}{(y-z)^{1.5}z^{1.5}}\right)\right) \nonumber \\
		&\leq O\left(\frac{c2^{2y-c}}{M(y-1,c-1)}\cdot\frac{1}{y^{1.5}} \left(\frac{e^{-\frac{c^2}{4y}}}{\sqrt{y-x}}+\frac{e^{-\frac{c^2}{4(y-x)}}}{\sqrt{y-x}}+\frac{e^{-\frac{c^2}{4(y-x)}}}{c}\right)\right) \nonumber \\
		&\leq O\left(\frac{c2^{2y-c}}{M(y-1,c-1)}\cdot\frac{e^{-\frac{c^2}{4y}}}{y^{1.5}} \left(\frac{1}{\sqrt{y-x}}+\frac{e^{-\frac{c^2}{8(y-x)}}}{c}\right)\right).\label{eqn_bound_c1}
	\end{align}

	The probability that $2L\leq s_c\leq U-3L$, $c$ is the first element in a interval with at most $l$ elements, is at most
	\begin{align*}
		&\sum_{y=2L}^{U-3L}\sum_{b=1}^{L}\frac{1}{L}\Pr[s_c=y, s_{c-1}\leq y-b, s_{c+l}>y-b+L]\\
		&=\sum_{y=2L}^{U-3L}\frac{\Pr[s_c=y]}{L}\sum_{b=1}^{L} \Pr[s_{c-1}\leq y-b,s_{c+l}>y-b+L\mid s_c=y]\\
		&\leq \sum_{y=2L}^{U-3L}\frac{M(y-1,c-1)M(U-y,m-c)}{L\cdot M(U,m)}\sum_{b=1}^{L}O\left(\frac{c2^{2y-c}}{M(y-1,c-1)}\cdot\frac{e^{-\frac{c^2}{4y}}}{y^{1.5}} \left(\frac{1}{\sqrt{b}}+\frac{e^{-\frac{c^2}{8b}}}{c}\right)\right. \\
		&\cdot\left. \frac{l(m-c-l)2^{2(U-y)-(m-c)}}{M(U-y,m-c)}\cdot \frac{e^{-\frac{(m-c)^2}{4(U-y)}}}{(U-y)^{1.5}}\cdot \left(\frac{1}{\sqrt{L-b}}+\frac{e^{-\frac{(m-c)^2}{32(L-b)}}}{m-c}\right)\right) \\
		&\leq O\left(\frac{2^{2U-m}cl(m-c-l)}{L\cdot M(U,m)}\sum_{y=1}^{U}\frac{e^{-\frac{c^2}{4y}-\frac{(m-c)^2}{4(U-y)}}}{y^{1.5}(U-y)^{1.5}}\sum_{b=1}^{L}\left(\frac{1}{\sqrt{b}}+\frac{e^{-\frac{c^2}{8b}}}{c}\right)\left(\frac{1}{\sqrt{L-b}}+\frac{e^{-\frac{(m-c)^2}{32(L-b)}}}{m-c}\right)\right) \\
		&\leq O\left(\frac{cl(m-c-l)U^{1.5}}{mL}\left(\sum_{y=1}^{U}\frac{e^{-\frac{c^2}{4y}-\frac{(m-c)^2}{4(U-y)}}}{y^{1.5}(U-y)^{1.5}}\right)\sum_{b=1}^{L}\left(\frac{1}{\sqrt{b}}+\frac{e^{-\frac{c^2}{8b}}}{c}\right)\left(\frac{1}{\sqrt{L-b}}+\frac{e^{-\frac{(m-c)^2}{32(L-b)}}}{m-c}\right)\right).
	\end{align*}
	By Fact~\ref{fact_sum3}, we have
	\begin{align*}
		\sum_{y=1}^{U}\frac{e^{-\frac{c^2}{4y}-\frac{(m-c)^2}{4(U-y)}}}{y^{1.5}(U-y)^{1.5}}&\leq O\left(U^{-1.5}e^{-(c^2+(m-c)^2)/2U}\cdot \left(1+\sqrt{\frac{c^2+(m-c)^2}{U}}+\sqrt{\frac{c^2(m-c)^2}{U(c^2+(m-c)^2}}\right)\cdot(1/c+1/(m-c))\right) \\
		&\leq O\left(U^{-1.5}(1/c+1/(m-c))\right).
	\end{align*}
	Next, we bound the last sum:
	\begin{align*}
		&\sum_{b=1}^{L}\left(\frac{1}{\sqrt{b}}+\frac{e^{-\frac{c^2}{8b}}}{c}\right)\left(\frac{1}{\sqrt{L-b}}+\frac{e^{-\frac{(m-c)^2}{32(L-b)}}}{m-c}\right) \\
		&\leq \sum_{b=1}^L\frac{1}{\sqrt{b(L-b)}}+\sum_{b=1}^L\frac{1}{(m-c)\sqrt{b}}+\sum_{b=1}^L\frac{1}{c\sqrt{L-b}}+\sum_{b=1}^Le^{-\frac{c^2}{8b}-\frac{(m-c)^2}{32(L-b)}} \\
		&\leq O\left(1+\frac{\sqrt{L}}{m-c}+\frac{\sqrt{L}}{c}\right)+L\cdot e^{-\frac{m^2}{32L}} \\
		&\leq O\left(1+\frac{\sqrt{L}}{m-c}+\frac{\sqrt{L}}{c}\right),
	\end{align*}
	where the last inequality is due to our assumption $L\log L=o(U)$.
	Hence, for $2\leq c\leq m-3l$, we have
	\begin{align*}
		&\Pr[2L\leq s_c\leq U-3L, \textrm{$c$ is the first element of a interval with at most $l$ elements}] \\
		&\sum_{y=2L}^{U-3L}\frac{1}{L}\sum_{b=1}^{L}\Pr[s_c=y, s_{c-1}\leq y-b, s_{c+l}>y-b+L]\\
		&\leq O\left(\frac{cl(m-c-l)}{mL}\left(\frac{1}{c}+\frac{1}{m-c}\right)\left(1+\frac{\sqrt{L}}{m-c}+\frac{\sqrt{L}}{c}\right)\right).
	\end{align*}
	Next, we take the sum over $c$:
	\begin{align*}
		&\sum_{c=2}^{m-3l} \frac{cl(m-c-l)}{mL}\left(\frac{1}{c}+\frac{1}{m-c}\right)\left(1+\frac{\sqrt{L}}{m-c}+\frac{\sqrt{L}}{c}\right) \\
		&\leq 2\sum_{c=1}^{m/2} \frac{cl(m-c)}{mL}\cdot \frac{2}{c}\cdot \left(1+\frac{2\sqrt{L}}{c}\right) \\
		&\leq O\left(\sum_{c=1}^{m/2}\frac{l}{L}\left(1+\frac{\sqrt{L}}{c}\right)\right) \\
		&\leq O\left(lm/L+(l\log m)/\sqrt{L}\right) \\
		&=O(lk/m),
	\end{align*}
	where the last inequality is due to our assumption $L=o(U/\log U)$.

	So far, we obtained an upper bound on the expected number of non-empty intervals that
	\begin{itemize}
		\item have at most $l$ elements;
		\item begin with $s_c$ such that $2L\leq s_c\leq U-3L$ and $2\leq c\leq m-3l$.
	\end{itemize}
	However, there could be at most $O(1)$ intervals that begin with an $s_c$ with $s_c\leq 2L$ or $s_c\geq U-3L$, or $c=1$.
	For $c>m-3l$, there could be at least $O(1)$ intervals with at least $l/2$ elements.

	Hence, the expected number of non-empty intervals with at most $l$ and at least $l/2$ elements is $O(lk/m+1)$.
\end{proof}


Let $S'\subseteq S$ be an (arbitrarily) jointly distributed subset of $S$, and $\E[|S'|]=(1-\epsilon)m$.
The above bound implies an upper bound on the number of non-empty intervals that do not contain any point in $S'$.
\begin{corollary}\label{cor_unk_int}
    Let $S'\subseteq S$ be an (arbitrarily) jointly distributed subset of $S$, and $\E[|S'|]=(1-\epsilon)m$.
    The number of non-empty intervals that do no contain any point in $S'$ has expectation at most $O(\sqrt{\epsilon k}+\log m)$, where the expectation is taken over the joint distribution of $S$, $\Delta$ and $S'$.
\end{corollary}
\begin{proof}
    By Lemma~\ref{lem_small_interval}, the number of intervals that have at most $m\sqrt{\epsilon /k}$ elements is at most $O(\sqrt{\epsilon k}+\log m)$ (regardless of whether it contains a point in $S'$).
    On the other hand, the expected number of intervals that have at least $m\sqrt{\epsilon /k}$ elements that do not contain any point in $S'$ is at most
    \[
        \frac{\epsilon m}{m\sqrt{\epsilon/k}}=\sqrt{\epsilon k}.
    \]
    Summing up the two parts proves the corollary.
\end{proof}

Next lemma considers the gaps between adjacent non-empty intervals.
It asserts that most non-empty intervals are consecutive intervals, and very few pairs are far away.

\begin{lemma}\label{lem_large_gap}
	The expected number of non-empty interval pairs that have between $t-1$ and $2t$ empty intervals and no non-empty interval in between is at most $O(\sqrt{k/t}+1)$.
	In particular, the expected number of non-empty intervals is at most $O(\sqrt{k}+\log K)$.
\end{lemma}
\begin{proof}
	Let us first bound the expected number of non-empty interval pairs with exactly $t-1$ empty intervals (and no non-empty interval) in between, i.e., adjacent non-empty interval pairs $(\mathcal{B}_1,\mathcal{B}_2)$ that are $t$ intervals far from each other.
	Instead of upper-bounding it directly, we are going to estimate for each element $s_c$, the probability that it is the last element in $\mathcal{B}_1$ and $s_{c+1}$ is the first element in $\mathcal{B}_2$.
	Then by linearity of expectation, taking the sum over $c$ gives us the desired bound.
	Indeed, for $1\leq c\leq m-1$ and $(t-1)L<y-x<(t+1)L$, the probability that $s_c=x,s_{c+1}=y$ and they form such a pair is at most
	\begin{align*}
		&\Pr[s_c=x,s_{c+1}=y]\cdot \left(1-\frac{|y-x-tL|}{L}\right) \\
		&=\frac{M(x-1,c-1)M(U-y,m-c)C_{y-x-1}}{M(U,m)}\cdot \left(1-\frac{|y-x-tL|}{L}\right) \\
		&\leq O\left(\frac{c(m-c)\cdot U^{1.5}}{m\cdot x^{1.5}(y-x)^{1.5}(U-y)^{1.5}}\cdot\frac{L-|y-x-tL|}{L}\right)
	\end{align*}
	For $x\leq U-5tL$, we have
	\begin{align*}
		&\sum_{y=x+(t-1)L+1}^{x+(t+1)L-1}\frac{c(m-c)\cdot U^{1.5}}{m\cdot x^{1.5}(y-x)^{1.5}(U-y)^{1.5}}\cdot\frac{L-|y-x-tL|}{L} \\
		&\leq O\left(\sum_{y-x=(t-1)L+1}^{tL}\frac{c(m-c)\cdot U^{1.5}}{m\cdot x^{1.5}(y-x)^{1.5}(U-x)^{1.5}} \left(\frac{y-x}{L}-(t-1)\right)\right) \\
		&=O\left(\frac{c(m-c)\cdot U^{1.5}}{m\cdot x^{1.5}(U-x)^{1.5}} \cdot \frac{1}{t^{1.5}\sqrt{L}}\right).
	\end{align*}
	Now, we take the sum over $x$, the expected number of such $(\mathcal{B}_1,\mathcal{B}_2)$ pairs where the last element in $\mathcal{B}_1$ is $s_c\leq U-5tL$, is at most
	\begin{align*}
		\sum_{x=1}^{U-5tL}O\left(\frac{c(m-c)\cdot U^{1.5}}{m\cdot x^{1.5}(U-x)^{1.5}} \cdot \frac{1}{t^{1.5}\sqrt{L}}\right)\leq O\left(\frac{c(m-c)}{m\cdot t^{1.5}\sqrt{L}} \right).
	\end{align*}
	Next, we take the sum over $c$,
	\[
		\sum_{c=1}^{m-1}O\left(\frac{c(m-c)}{m\cdot t^{1.5}\sqrt{L}} \right)=O\left(\frac{m}{t^{1.5}\sqrt{L}}\right)=O(\sqrt{k}/t^{1.5}).
	\]
	A similar argument shows that, for any $t'\in [t-1,2t)$, the expected number of adjacent non-empty interval pairs $(\mathcal{B}_1,\mathcal{B}_2)$ such that
	\begin{itemize}
	 	\item they are $t'$ intervals far, and
	 	\item the last element in $\mathcal{B}_1$ is at most $U-5tL$
	 \end{itemize}
	is at most $O(\sqrt{k}/t'^{1.5})$.
	However, there could be at most $O(1)$ interval pairs that are $[t-1,2t)$ intervals far after $U-5tL$.
	Hence, taking the sum over $t'$, the expected total number of such pairs is $O(\sqrt{k/t}+1)$.

	In particular, taking the sum over $t=2^i$ for $i=0,\ldots,\log K$ gives us an upper bound of $O(\sqrt{k}+\log K)$ on the expected number of non-empty pairs.
\end{proof}


Now, we are ready to prove Lemma~\ref{lem_enc_len}.
Note that $k,K=\Theta(\log^4 B)$.
\begin{restate}[Lemma~\ref{lem_enc_len}]
	For random $S$ and $\Delta$, let $S'\subseteq S$ be a subset (arbitrarily) jointly distributed.
	Then there is a prefix-free binary string $\com_{\Delta,S}(S')$, such that $\com_{\Delta,S}(S')$ and $S'$ together determine $\{E_j\}_{j\in [K]}$ (or equivalently, the set of non-empty intervals).
	Moreover, we have the following bound on the length of $\com_{\Delta,S}(S')$:
	\[
		\E_{\Delta,S,S'}\left[
		  \left|\com_{\Delta,S}(S')\right|
		\right]\leq 
		O\left(\sqrt{\epsilon} \log^2 B\log(1/\epsilon)+\log B\log\log B\right),
	\]
	where $\epsilon:=\E[1-|S'|/m]$.  
\end{restate}
\begin{proof}
	First observe that for each non-empty interval, when at least one of its elements appears in $S'$, we already know this interval is non-empty.
	In the other words, it suffices to encode in $\com_{\Delta,S}(S')$, the non-empty intervals that none of the elements appears in $S'$.

	\paragraph{Encode $\{E_j\}$ given $S'$.}
	We now describe $\com_{\Delta,S}(S')$, which encodes the non-empty intervals given $S'$.
	To this end, we first compute all non-empty intervals, as well as the non-empty intervals that contain no points in $S'$.
	Let $K_{ne}$ be the number of non-empty intervals, and $K_{un}$ be the number of ``unknown'' non-empty intervals (i.e. the non-empty intervals that do not contain a point in $S'$).
	Let $I_1,I_2,\ldots,I_{K_{ne}}\in[K]$ be all non-empty intervals in the increasing order.
	Let $I_{i_1}, I_{i_2},\ldots,$ be all non-empty intervals that contain no points in $S'$ (the ``unknown'' intervals to be encoded).
	We first write down $K_{ne}$, then for each ``unknown'' interval $I_{i_a}$, we do the following:
	\begin{enumerate}
	    \item write down $i_a-i_{a-1}$ ($i_0$ is assumed to be $0$);
	    \item write down $I_{i_a}-I_{i_a-1}$.
	\end{enumerate}
	All integers are encoded using the folklore prefix-free encoding which takes $O(\log N)$ bits to encode an integer $N$.
	This completes $\com_{\Delta,S}(S')$.

	\paragraph{Decode $\{E_j\}$ given $\com_{\Delta,S}(S')$ and $S'$.}
	To decode $\{E_j\}$, we first compute the list $\mathcal{J}$ of non-empty intervals that contain at least one point in $S'$.
	Next, we read $K_{ne}$ (which together with $|\mathcal{J}|$, determines $K_{un}$), and for $a=1,\ldots,K_{un}$, do the following:
	\begin{enumerate}
		\item read the next integer and recover $i_a$;
		\item for $i=i_{a-1}+1,\ldots,i_a-1$, let $I_i$ be the next interval in $\mathcal{J}$ (the intervals that do no require encoding);
		\item read the next integer and recover $I_{i_a}$.
	\end{enumerate}
	Finally, for $i=i_{K_{un}}+1,\ldots,K_{ne}$, let $I_i$ be the next interval in $\mathcal{J}$.
	This recovers all $I_1,\ldots,I_{K_{ne}}$, hence, decodes $\{E_j\}$.

	\paragraph{The length of $\com_{\Delta,S}(S')$.}
	Next, we analyze the expected length of $\com_{\Delta,S}(S')$.
	$K_{ne}$ takes $O(\log K)$ bits to encode.
	Then for $a=1,\ldots,K_{un}$, $i_a-i_{a-1}$ takes $O(\log (i_a-i_{a-1}))$ bits to encode.
	Since all these integers sum up to (at most) $K_{ne}$, by concavity of $\log$, the total number of bits used to encode $\{i_a-i_{a-1}\}$ is at most 
	\[
		O(K_{un}\cdot \log\frac{K_{ne}}{K_{un}}).
	\]
	By Corollary~\ref{cor_unk_int}, $\E[K_{un}]\leq O(\sqrt{\epsilon k}+\log m)=O(\max\{\sqrt{\epsilon k},\log m\})$.
	By Lemma~\ref{lem_large_gap}, $\E[K_{ne}]\leq O(\sqrt{k})$.
	Then by the concavity and monotonicity of $f(x, y)=x\ln(y/x)$, the expected encoding length of all $i_a-i_{a-1}$ is at most
	\begin{equation}\label{eqn_length_1}
		O\left(\max\{\sqrt{\epsilon k},\log m\}\log \left(\frac{\sqrt{k}}{\max\{\sqrt{\epsilon k},\log m\}}\right)\right).
	\end{equation}

	Next, the value $I_{i_a}-I_{i_a-1}$ takes $O(\log (I_{i_a}-I_{i_a-1}))$ bits to encode.
	For all ${I_{i_a}-I_{i_a-1}\leq \frac{\sqrt{k}}{\max\{\sqrt{\epsilon k},\log m\}}}$, their total encoding length is at most
	\[
		O\left(K_{un}\cdot \log \left(\frac{\sqrt{k}}{\max\{\sqrt{\epsilon k},\log m\}}\right)\right),
	\]
	and its expectation is at most
	\begin{equation}\label{eqn_length_2}
		O\left(\max\{\sqrt{\epsilon k},\log m\}\cdot \log \left(\frac{\sqrt{k}}{\max\{\sqrt{\epsilon k},\log m\}}\right)\right).
	\end{equation}

	For all $I_{i_a}-I_{i_a-1}>\frac{\sqrt{k}}{\max\{\sqrt{\epsilon k},\log m\}}$, by Lemma~\ref{lem_large_gap}, their expected encoding length is at most
	\begin{align}
		&\ \sum_{t=\frac{2^b\cdot \sqrt{k}}{\max\{\sqrt{\epsilon k},\log m\}}:b\geq 0, t\leq K} O((\sqrt{k/t}+1)\log t)\nonumber\\
		\leq&\ O\left(\max\{\sqrt{\epsilon k},\log m\}\cdot\log \left(\frac{\sqrt{k}}{\max\{\sqrt{\epsilon k},\log m\}}\right)+\log^2 K\right).\label{eqn_length_3}
	\end{align}

	Finally, summing up \eqref{eqn_length_1}, \eqref{eqn_length_2} and \eqref{eqn_length_3}, the expected length of $\com_{S,\Delta}(S')$ is at most
	\[
		O\left(\max\{\sqrt{\epsilon k},\log m\}\cdot\log \left(\frac{\sqrt{k}}{\max\{\sqrt{\epsilon k},\log m\}}\right)+\log^2 K\right).
	\]
	Since $k, K=\Theta(\log^4 B)$, $m\leq B$ and $f(x, y)=x\log (y/x)$ is non-decreasing when $x\leq y/e$, it is at most
	\[
		O\left(\max\{\sqrt{\epsilon k}\log (1/\sqrt{\epsilon}), \log B\cdot \log\log B\}+\log^2\log B\right)\leq O\left(\sqrt{\epsilon k}\log (1/\epsilon)+\log B\cdot \log\log B\right).
	\]
	This proves the lemma.
\end{proof}

Finally, we prove Lemma~\ref{lem_entropy_lb} (and note that $k=\log^4 B$).
\begin{restate}[Lemma~\ref{lem_entropy_lb}]
	The entropy of $\{E_j\}_{j\in[K]}$ is at least $\Omega(\sqrt{k})$.
\end{restate}
\begin{proof}



	To prove the entropy lower bound, we will apply the chain rule, and show that for many $j$, $H(E_{j-1}\mid E_{\geq j})$ is large.
	Denote by $c$, the smallest element in $E_{\geq j}$.
	Then $c$ and $s_c,\ldots,s_m$ determine $E_{\geq j}$.
	Thus, it suffices to lower bound $H(E_{j-1}\mid c,s_c,\ldots,s_m)$, as conditioning on more variables could only decrease the entropy.

	To this end, fix $j\in [k/30, 2k/30]$, and suppose the $j$-th interval is $[x, x+L)$.
	Then $x\in [m^2/30, 2m^2/30]$, and $x<2U/3$.
	We denote by $W_j$, the event that $c\in[m/3, 2m/3]$ and $s_c\in [x+L/2, x+L)$.
	The entropy lower bound follows from the following two claims.
	\begin{claim}\label{cl_prob_lb}
		The probability of $W_j$ is at least $\Pr[W_j]\geq \Omega(1/\sqrt{k})$.
	\end{claim}
	\begin{claim}\label{cl_ce_lb}
		The conditional entropy of $E_{j-1}$ conditioned on $W_j$, is at least 
		\[
			H(E_{j-1}\mid W_j,c,s_c,\ldots,s_m)\geq \Omega(1).
		\]
	\end{claim}

	We first prove the lemma assuming the two claims.
	Since $W_j$ is an event that depends only on $c$ and $s_c$, by definition, we have
	\begin{align*}
		H(E_{j-1}\mid E_{\geq j})&\geq H(E_{j-1}\mid c,s_c,\ldots,s_m)\\
		&=\E_{c_0,z}[H(E_{j-1}\mid c=c_0,s_c=z,\ldots,s_m)] \\
		&\geq \Pr[W_j]\cdot \E_{c_0,z\mid W_j}[H(E_{j-1}\mid W_j,c=c_0,s_c=z,\ldots,s_m)] \\
		&\geq \Pr[W_j]\cdot H(E_{j-1}\mid W_j,c,s_c,\ldots,s_m) \\
		&\geq \Omega(1/\sqrt{k}).
	\end{align*}
	Finally, by chain rule, we have
	\begin{align*}
		H(E_1,\ldots,E_K)&=\sum_{j=1}^{K}H(E_j\mid E_{j+1},\ldots,E_K) \\
		&\geq \sum_{j=k/30}^{2k/30} H(E_{j-1}\mid E_j,\ldots,E_K) \\
		&\geq \Omega(\sqrt{k}).
	\end{align*}
	This proves the lemma.
	Hence, it suffices to prove the two claims.

	\bigskip
	To prove Claim~\ref{cl_prob_lb}, let us first lower bound the probability that $s_c=z$ for some $z\in[x+L/2,x+L)$,
	\begin{align*}
		\Pr[s_c=z]&=\sum_{y=1}^{x}\Pr[s_c=z,s_{c-1}=y] \\
		&\geq \sum_{y=x-L}^x \frac{M(y-1,c-2)C_{z-y-1}M(U-z,m-c)}{M(U,m)} \\
		\intertext{which by Proposition~\ref{prop_M}, is at least}
		&\geq \Omega\left(\sum_{y=x-L}^x \frac{c(m-c)U^{1.5}e^{-O(c^2/y+(m-c)^2/(U-z))+\Omega(m^2/U)}}{y^{1.5}(z-y)^{1.5}(U-z)^{1.5}m}\right) \\
		&= \Omega\left(\sum_{y=x-L}^x\frac{m}{x^{1.5}(z-y)^{1.5}}\right) \\
		&\geq \Omega\left(\frac{1}{m^2\sqrt{L}}\right).
	\end{align*}
	Now, we take the sum over $z$ from $x+L/2$ to $x+L$ and over $c$ from $m/3$ to $2m/3$, proving $\Pr[W_j]=\sum_{c\in[m/3,2m/3]}\sum_{z\in[x+L/2,x+l)}\Pr[s_c=z]\geq\Omega(\sqrt{L}/m)=\Omega(1/\sqrt{k})$.

	\bigskip
	To prove Claim~\ref{cl_ce_lb}, it suffices to show that the conditional probability of $E_{j-1}=0$ is bounded away (by a constant) from both $0$ and $1$.
	Note that $s_c$ is the first element after $(j-1)$-th interval, $s_{c-1}$ must be in $(j-1)$-th interval or earlier.
	We have
	\begin{align*}
		\Pr[E_{j-1}=0\mid c,s_c=z,s_{c+1},\ldots,s_m]&=\Pr[s_{c-1}< x-L\mid s_c=z,s_{c-1}< x] \\
		&=\frac{\Pr[s_{c-1}<x-L\mid s_c=z]}{\Pr[s_{c-1}<x\mid s_c=z]} \\
		&=\frac{\sum_{y=1}^{x-L-1}M(y-1,c-2)C_{z-y-1}}{\sum_{y=1}^{x-1}M(y-1,c-2)C_{z-y-1}}
	\end{align*}
	Recall that $c=\Theta(m),x=\Theta(m^2),z-x=\Theta(L)$.
	Note that $M(y,c)C_{z-y}$ is increasing when $y=O(m^2)$ and $\frac{\exp(-c^2/4y)}{\exp(-c^2/(4y-2c))}=1+o(1)$ when $y^2=\omega(c^3)$, thus the probability is
	\[
		=(1\pm o(1))\frac{\sum_{y=1}^{x-L}e^{-c^2/4y}y^{-1.5}(z-y)^{-1.5}}{\sum_{y=1}^{x}e^{-c^2/4y}y^{-1.5}(z-y)^{-1.5}}.
	\]

	On the one hand, it is at least
	\begin{align*}
		\frac{\sum_{y=1}^{x-L}e^{-c^2/4y}y^{-1.5}(z-y)^{-1.5}}{\sum_{y=1}^{x}e^{-c^2/4y}y^{-1.5}(z-y)^{-1.5}}&\geq \frac{\sum_{y=x-2L}^{x-L}e^{-c^2/4(x-2L)}(x-L)^{-1.5}(z-(x-2L))^{-1.5}}{2\sum_{y=x/2}^{x}e^{-c^2/4y}y^{-1.5}(z-y)^{-1.5}} \\
		&\geq \Omega\left(\frac{\sum_{y=x-2L}^{x-L}x^{-1.5}(z-x)^{-1.5}}{\sum_{y=x/2}^{x}x^{-1.5}(z-y)^{-1.5}}\right) \\
		&=\Omega\left(\frac{L\cdot x^{-1.5}\cdot L^{-1.5}}{x^{-1.5}L^{-0.5}}\right) \\
		&=\Omega(1),
	\end{align*}
	since $z-x\in [L/2,L)$, $L\ll x$ and $c^2/x=O(1)$.

	On the other hand, it is at most
	\begin{align*}
		\frac{\sum_{y=1}^{x-L}e^{-c^2/4y}y^{-1.5}(z-y)^{-1.5}}{\sum_{y=1}^{x}e^{-c^2/4y}y^{-1.5}(z-y)^{-1.5}}&=1-\frac{\sum_{y=x-L+1}^{x}e^{-c^2/4y}y^{-1.5}(z-y)^{-1.5}}{\sum_{y=1}^{x}e^{-c^2/4y}y^{-1.5}(z-y)^{-1.5}} \\
		&\leq 1-\Omega\left(\frac{\sum_{y=x-L+1}^{x}e^{-c^2/4(x-L)}x^{-1.5}(z-y)^{-1.5}}{x^{-1.5}L^{-0.5}}\right) \\
		&\leq 1-\Omega\left(\frac{L(z-(x-L))^{-1.5}}{L^{-0.5}}\right) \\
		&\leq 1-\Omega(1).
	\end{align*}

	In the other words, $\Pr[E_{j-1}=0\mid c,s_c=z,s_{c+1},\ldots,s_m]$ is always bounded away from both $0$ and $1$.
	Thus, the conditional entropy is at least a constant
	\[
		H(E_{j-1}\mid c,s_c=z,s_{c+1},\ldots,s_m)\geq \Omega(1).
	\]
	This bound holds for all $c,z$ that satisfy $W_j$, hence,
	\[
		H(E_{j-1}\mid W_j, c,s_c,s_{c+1},\ldots,s_m)\geq \Omega(1).
	\]
	This completes the proof of the lemma.
\end{proof}

\bibliography{refs}
\bibliographystyle{alpha}

\appendix
\ifx\mainfile\undefined
\documentclass{article}

\allowdisplaybreaks

\begin{document}
\fi

\section{Proofs of the Facts}

%

\begin{proof}[Proof of Fact~\ref{fact_sum1}]
Observe that $\exp(-A/h)h^{-1.5}$ is increasing when $h\in[1,1.5A]$ and decreasing when $h>1.5A$.
Thus if $T\le 1.5A$,
\begin{align*}
    \sum_{h=1}^T\exp(-A/h)h^{-1.5}
    \le\exp(-A/T)\sum_{h=1}^TT^{-1.5}
    =\exp(-A/T)/\sqrt{T}.
\end{align*}
Otherwise,
\begin{align*}
    \sum_{h=1}^T\exp(-A/h)h^{-1.5}
    \le \sum_{h=1}^{1.5A}\exp(-A/h)h^{-1.5}+O(1/\sqrt{A})
    =O(1/\sqrt{A}).
\end{align*}
Then the desired inequality follows.
\end{proof}

\begin{proof}[Proof of Fact~\ref{fact_sum2}]
\begin{align*}
    \sum_{h=1}^T \exp(-A/(T-h))\cdot h^{-0.5}
    \le\exp(-A/T)\sum_{h=1}^T \cdot h^{-0.5}
    =O(\exp(-A/T)\int_1^T\frac{dx}{\sqrt{x}})
    =O(\exp(-A/T)\sqrt{T}).
\end{align*}
\end{proof}

\begin{proof}[Proof of Fact~\ref{fact_sum3}]
Without loss of generality, assume $A_1\le A_2$.
Let $h^*:= A_1T/(A_1+A_2)$ the point such that $A_1/h^*=A_2/(T-h)$.
Note that $h^*\le T/2$.
Thus
\begin{align*}
    &\sum_{h=1}^{T-1}\exp(-A_1/h-A_2/(T-h))h^{-1.5}(T-h)^{-1.5}\\
    \le& O\left(\sum_{h=1}^{h^*}\exp(-A_1/h)h^{-1.5}(T-h)^{-1.5}+\sum_{h=h^*}^{T}\exp(-A_2/(T-h))h^{-1.5}(T-h)^{-1.5}\right)
\end{align*}
The former term is upper bounded by
\begin{align*}
    T^{-1.5}\sum_{h=1}^{h^*}\exp(-A_1/h)h^{-1.5}
    \le T^{-1.5}\exp(-A_1/h^*)(1/\sqrt{h^*}+1/\sqrt{A_1}),
\end{align*}
due to the fact $h^*<T/2$ and Fact~\ref{fact_sum1}.
Note that
\begin{align*}
    \sum_{h=h^*}^{T/2}\exp(-A_2/(T-h))h^{-1.5}(T-h)^{-1.5}
    \le O(T^{-1.5}\exp(-2A_2/T)(1/\sqrt{h^*}-1/\sqrt{T}).\\
    T^{-1.5}\sum_{h=T/2}^T\exp(-A_2/h)(T-h)^{1.5}
    \le O(T^{-1.5}\exp(-A_2/T)(1/\sqrt{T}+1/\sqrt{A_2})).
\end{align*}
Thus the latter term is upper bounded by
\[
    T^{-1.5}\exp(-A_2/T)(1/\sqrt{h^*}+1/\sqrt{A_2}+1/\sqrt{T}).
\]
By the symmetry of the expression,
\begin{align*}
    &\sum_{h=1}^{T-1}\exp(-A_1/h-A_2/(T-h))h^{-1.5}(T-h)^{-1.5}\\
    \le& O\left(T^{-1.5}\exp(-(A_1+A_2)/2T)(\sqrt{\frac{A_1+A_2}{A_1T}}+\sqrt{\frac{A_1+A_2}{A_2T}}+1/\sqrt{A_1}+1/\sqrt{A_2}+1/\sqrt{T})\right)\\
    =& O\left(T^{-1.5}\exp(-(A_1+A_2)/2T)(1+\sqrt{\frac{A_1+A_2}{T}}+\sqrt{\frac{A_1A_2}{T(A_1+A_2)}})(1/\sqrt{A_1}+1/\sqrt{A_2})\right)\\
\end{align*}
\end{proof}

\begin{proof}[Proof of Fact~\ref{fact_conv}]
The following inequality is an alternative form of the desired inequality:
\[
    8(\frac{A_2}{B_2}-\frac{A_1}{B_1})^2\ge\frac{A_2^2}{B_2}(\frac{1}{B_1}+\frac{1}{B_2}).
\]
Recall that $A_2>2A_1$ and $B_1\ge B_2$, so the L.H.S. is at least $2(A_2/B_2)^2$ and the R.H.S. is at most $2(A_2/B_2)^2$.
\end{proof}


\ifx\mainfile\undefined
\end{document}
\fi

\ifx\mainfile\undefined
\documentclass{article}

\begin{document}
\fi

%



\section{A Constant Redundancy Algorithm}\label{app_upper_one}
In this section, we aim to propose an algorithm which solves RMQ using space of $\log_2C_n+1$ bits and answer any query in $O(\log n)$ times.
The upper bounds from \cite{NS14} imply exactly identical result, but our approach is totally different from \cite{NS14} which uses the algorithm from \cite{FH11} as a black-box.
Our approach is much simpler, clearer and easier to be implemented.

Our algorithm is some kind of augmented binary (search) tree called by \mihai\cite{patracscu2008lower,Pat08} or segment tree called by Chinese competitive programming participants.
Note that the structure is different with the one used to store segments invented by Jon Louis Bentley\cite{bentley1977solutions}.
\begin{itemize}
  \item 
	The data structure is a binary tree which represents some properties of an array $A[1\dots n]$.
	The value of $i$-th leaf of the in-order tree traversal is a function of $A[i]$.
  \item 
	The tree is a recursive structure.
	Every internal node equips a tiny structure to aid query algorithm.
	Let the leaves contained in a subtree $v$ be the $i$-th to $j$-th leaves of the in-order tree traversal.
	Then the root node of subtree $v$ represents the subarray $A[i\dots j]$.
	The tree is a complete binary tree, thus there are roughly $\sum_{i=0}^{\log_2n} n/2^i\le 2n$ nodes.
  \item
	Usually, given any query on subarray $A[i\dots j]$, 
	the query will be broken into at most $O(\log n)$ queries on smaller subarraies represented by the nodes of the tree.
	But here in our algorithm, the query algorithm will return the answer upon the query interval is broken, therefore we will access at most $O(\log n)$ nodes in the tree.
\end{itemize}

\paragraph{Query Algorithm.}
Given a query $[a,b]$, starting from the root of the tree, we can always recurse to one of the children if $[a,b]$ is completely covered by the interval represented by the child.
The problem is how do we find the position of the minimum value in the range $[a,b]$ if $[a,b]$ overlaps with both of the two children intervals.
Given a node $v$, let $[L_v,R_v]$ be the interval represented by $v$, let $S_{l,v},S_{r,v}\in[n]^*$ be the sequences of positions of the greedy longest decreasing subsequences starting from $L_v$ 
in $A[L_v,R_v]$ and starting from $R_v$ in $A[R_v,L_v]$.
Formally, $S_{l,v,1}\triangleq L_v, S_{l,v,i}\triangleq\min\{j\in[S_{l,v,i-1},R_V]:A[j]<A[S_{l,v,i-1}]\}$ and $S_{r,v,1}\triangleq R_V, S_{r,v,i}\triangleq\max\{j\in[L_v,S_{r,v,i-1}]:A[j]<A[S_{r,v,i-1}]\}$.
Suppose our query $[a,b]$ is covered by a pair of sibling nodes $p,q$ with paramenters $a\le R_p, b\ge L_q$.
Our idea is that the answer for query $[a,b]$ obviously is $\argmin_{i\in\{S_{r,p,a'},S_{l,q,b'}\}}A[i]$ if we have $a\in(S_{p,r,a'+1},S_{p,r,a'}]$ and $b\in[S_{q,l,b'},S_{q,l,b'+1})$.
The remaining problems for our query algorithm are 
\begin{enumerate}
  \item how to find $a'$ and $b'$?
  \item how to compare $A[S_{r,p,a'}]$ and $A[S_{l,q,b'}]$?
\end{enumerate}
To solve the first problem, we store two values $|S_{l,v}|$ and $|S_{r,v}|$ in any node $v$.
\begin{claim}
  The following algorithm find $a',b'$  in $O(\log n)$ times.
\end{claim}
\begin{algorithm}[H]\label{relocate algorithm}
  \SetKwProg{Fn}{Function}{ is}{end}

  \caption{algorithm to relocate the range}
  \SetKwFunction{FindA}{FindA}
  \SetKwFunction{FindB}{FindB}
  \SetKwFunction{lchild}{LChild}
  \SetKwFunction{rchild}{RChild}
  \SetKwInOut{Input}{input}\SetKwInOut{Output}{output}
  \SetKwData{lson}{LChild}
  \SetKwData{rson}{RChild}
  \SetKwData{offset}{offset}
  \SetKwData{offseta}{offsetA}
  \SetKwData{offsetb}{offsetB}

  \Input{query $[a,b]$; a pair of sibling nodes $p,q$ with paramenters $a\le R_p$ and $b\ge L_q$.}
  \Output{a tuple $(a',b')$ such that $a'=\min\{i:S_{p,r,i}\ge a\}$ and $b'=\min\{i:S_{q,l,i}\le b\}$.}
%
  \Return{(\FindA{$p$},\FindB{$q$})}

  \Fn{\FindA{current node: $v$}}
  {
	\lIf(\tcp*[f]{It is a leaf node}){$L_v=R_v$} { \Return 1 }
	$\lson\gets\lchild{v}$\;
	$\rson\gets\rchild{v}$\;
	\lIf(\tcp*[f]{In the right subtree}){$a\ge L_{\rson}$} { \Return \FindA{\rson} }
	\tcp{Otherwise $a$ is in the left subtree}
	$\offset\gets$\FindA{\lson}\;
	\Return $\max\{0,\offset-(|S_{r,\lson}|+|S_{r,\rson}|-|S_{r,v}|)\}+|S_{r,\rson}|$\;
  }
  \Fn{\FindB{current node: $v$}}
  {
	\lIf(\tcp*[f]{It is a leaf node}){$L_v=R_v$} { \Return 1 }
	$\lson\gets\lchild{v}$\;
	$\rson\gets\rchild{v}$\;
	\lIf(\tcp*[f]{In the left subtree}){$a\le R_{\lson}$} { \Return \FindB{\lson} }
	\tcp{Otherwise $a$ is in the rigt subtree}
	$\offset\gets$\FindB{\rson}\;
	\Return $\max\{0,\offset-(|S_{l,\lson}|+|S_{l,\rson}|-|S_{l,v}|)\}+|S_{l,\lson}|$\;
  }
\end{algorithm}
Consider the process when we merge the two array $A[S_{r,p}]$ and $A[S_{l,q}]$ into a sorted decreasing array $A'$.
We define a boolean string $\mathrm{Merge}_{pq}\in \HammingCube{|S_{r,p}|+|S_{r,q}|}$ as the witness of the process: $\mathrm{Merge}_{pq,i}=0\iff A'[i]$ is from $A[S_{r,p}]$, $\mathrm{Merge}_{pq,i}=1\iff A'[i]$ is from $A[S_{l,q}]$.
It is easy to see that, to compare $A[S_{r,p,a'}]$ and $A[S_{l,q,b'}]$, it is sufficient to compare the index of $a'$-th $0$ and the index of $b'$-th $1$.
To this end, we maintain an additional data structure in every internal node such that any \emph{select} query can be answered in time $O(\log|\mathrm{Merge}_{pq}|)$ for any internal node with children nodes $p,q$.
The additional data structure is a fundamental application of augmented binary (tree) or segment tree\cite{patracscu2008lower}.
Let $\mathrm{Merge}_{pq,(i)}$ denote the position of $i$-th smallest one in $\mathrm{Merge}_{pq}$.
We just construct a string of numbers $T\in [u]^n$ from $\mathrm{Merge}_{pq}$ in the following way: $T_1\triangleq \mathrm{Merge}_{pq,(1)},\forall i>1, T_i\triangleq \mathrm{Merge}_{pq,(i)}-\mathrm{Merge}_{pq,(i-1)}$.
Then the \emph{select}$(i)$ on $\mathrm{Merge}_{pq}$ is equal to \emph{rank}$(i)$ on $T$.
It is easy to see that 
\begin{align}
  A[S_{r,p,a'}]<A[S_{l,q,b'}]\iff\mathit{select}(b',\mathrm{Merge}_{pq})=\mathit{rank}(b',T)\ge a'+b'.
  \label{comparing by selecting}
\end{align}

\paragraph{Table Construction.}
The remaining part is construct a data structure meets everything the query algorithm needs with at most $\log C_n +1$ bits, where $C_n=\binom{2n}{n}/(n+1)\sim\frac{2^{2n}}{n^{3/2}\sqrt{\pi}}$ is the $n$-th Catalan number.
Note that the number of possible databases of length $n$ for RMQ exactly is the number of binary trees of size $n$, i.e. the $n$-th Catalan number.
To this end, we represent our data structure in spill-over representation by applying \mihai's technique\cite{Pat08}.
In spill-over representation, an element $x\in \mathcal{X}$ is represented as a tuple $(y_m,y_k)\in\HammingCube{m}\times[k]$ for some integers $m,k$, so the number of bits used to store $x$ is considered as $m+\log_2 k$.
We want to apply \mihai's lemma in a black-box way.
\begin{lemma}[Lemma 5 from \cite{Pat08}]
  \label{compression lemma}
  Assume we have to represent a variable $x\in\mathcal{X}$, and a pair $(y_M,y_K)\in\HammingCube{M(x)}\times\{0,\dots,K(x)-1\}$.
  Let $p:\mathcal{X}\to\mathbb{R}$ be a probability density funciton on $\mathcal{X}$, and $K,M:\mathcal{X}\to\mathbb{N}$ be non-negative functions on $\mathcal{X}$ satisfying:
  \begin{align}
	\forall x\in\mathcal{X}: \log_2\frac{1}{p(x)}+M(x)+\log_2K(x)\le H
	\label{spillover space upper bound}
  \end{align}
  We can design a spill-over representation of $x,y_M$ and $y_K$ with the following parameters:
  \begin{itemize}
	\item the spill universe is $K_\star$ with $K_\star\le 2r$, and the memory usage is $M_\star$ bits;
	\item the redundancy is at most $4/r$ bits, i.e. $M_\star+\log_2K_\star\le H+4/r$;
	\item if the word size is $w=\Omega(\log|\mathcal{X}|+\log r+\log \max_x K(x))$, $x$ and $y_K$ can be decoded with $O(1)$ word probes.
	  The input bits $y_M$ can be read directly from memory, but only after $y_K$ is retrieved;
	\item given a precomputed table of $O(|\mathcal{X}|)$ words that only depends on the input functions $K,M$ and $p$, decoding $x$ and $y_K$ takes constant time on the world of RAM.
  \end{itemize}
\end{lemma}
However, the issue is that our tiny structures equiped by internal nodes can not be considered as a part of varable $x$ in Lemma (\ref{compression lemma}), since the cardinality of the universe of the structure can be as large as $\exp(\Omega(n))$, which can not be read in $O(1)$ times by a cell-probe shceme with word size $w=O(\log n)$.
To fix this issue, we treat the tiny structure as the third child of the internal nodes.
Assume all the tiny structures are prepared into a spill-over representation in following way:
\begin{claim}\label{select algorithm}
  We can design a spill-over representation of any $S\in\binom{[u]}{n}$ with the following paramenters:
  \begin{itemize}
	\item the spill-over universe is $K_S\le 2r$, and the memory usage is $M_S$ bits;
	\item the redundancy is at most $4/r$  bits, i.e. $M_S+\log_2 K_S\le \log_2\binom{u}{n}+4/r$;
	\item if the word size is $w=\Omega(\log u+\log r)$, for any $i\in[n]$, the query \emph{select}$(i)$ can be answered in time $O(\log n)$.
  \end{itemize}
\end{claim}
The algorithmm is a simple application of \mihai's algorithm \cite{Pat08}, we omit the proof here.
Now we are ready to show the constrution of our table.

Let $\varphi$ be some label assigned to some node, and $\varphi_v$ a label assigned to node $v$.
For any leaf $v$, the universe of $\varphi_v$ will be identical, and the cardinality of the universe is $1$.
Hence there will be nothing to be stored in a leaf.
For a internal node $v$, $\varphi_v\triangleq(|S_{l,v}|,|S_{r,v}|)$ is a tuple.
Let $\mathcal{N}(n,\varphi)$ denote the number of possible instances of cartesian trees of $A[1\dots n]$ conditioning on the root of our data structure is labeled with $\varphi$.
Note that we can write the following recursion of $\mathcal{N}(a+,b,\varphi)$:
\[
  \mathcal{N}(a+b,\varphi)=\sum_{\varphi',\varphi'':\mathcal{A}(\varphi',\varphi'')\ni\varphi}\mathcal{N}(a,\varphi')\cdot\mathcal{N}(b,\varphi'')\cdot\mathcal{M}(\varphi,\varphi',\varphi''),
\]
where $\mathcal{A}(\varphi',\varphi'')$ is the set of labels can be merged into from children with labels $\varphi'$ and $\varphi''$, $\mathcal{M}(\varphi,\varphi',\varphi'')$ is the number of possible merge witnesses given the event that children is labled with $\varphi',\varphi''$ and parent node is labeled with $\varphi$.
In particular,
\[
  \mathcal{A}(\varphi',\varphi'')\triangleq\{(l,r): (\varphi'_l=l\land \varphi'_r+\varphi''_r\ge r\ge \varphi''_r+1)\lor(\varphi''_r=r\land\varphi'_l+\varphi''_l\ge l\ge \varphi'_r+1)\},
\]
where $\varphi_l,\varphi_r$ are the first and second elements of $\varphi$ respectively.
Observe that the two substrees and the merge witness are mutual independent if the three labels $\varphi,\varphi',\varphi''$ are fixed.
Also observe that $\mathcal{M}(\varphi,\varphi',\varphi'')$ is a binomial coefficent but a little bit complicated:
\begin{align}
  \mathcal{M}(\varphi,\varphi',\varphi'')\triangleq
  \begin{cases}
	\binom{\varphi'_r-(\varphi_r-\varphi''_r)+\varphi''_l}{\varphi''_l}	&	\text{(if $\varphi_l=\varphi'_l$)}\\
	\binom{\varphi''_l-(\varphi_l-\varphi'_l)+\varphi'_r}{\varphi'_r}	&	\text{(if $\varphi_r=\varphi''_r$)}
  \end{cases}
  \label{def-M}
\end{align}

Assume all the tiny structures are prepared into a spill-over representation according to Claim (\ref{select algorithm}) with the binomial coefficent in Eq(\ref{def-M}).
Let $K(n,\varphi), M(n,\varphi)$ be the spill universe and the memory bits used by our spill-over representation for any input array of length $n$ and root label $\varphi$.
Let $r$ to be determined.
We guarantee inductively that:
\begin{align}
  K(n,\varphi)&\le 2r;\label{spill-over k upper bound}
  \\
  M(n,\varphi)+\log_2K(n,\varphi)&\le \log_2\mathcal{N}(n,\varphi)+8\cdot\frac{n-1}{r}.\label{spill-over mk upper bound}
\end{align}
For a leaf, there are nothing to be stored.
So $K(1,*)=1, M(1,*)=0$.
For a internal node $v$, we assume the array of length $n$ is broken into two subarrays of length $a$ and $b$ without loss of generality.
Let $\varphi',\varphi''$ be labels of the children of $v$.
We recursively construct data structures for both of the two subtrees using space $(M(a,\varphi'),K(a,\varphi'))$ and $(M(b,\varphi''),K(b,\varphi''))$ respectively.
We also construct our tiny structure for $v$ using space $K(\varphi,\varphi',\varphi'')\le 2r$ and $M(\varphi,\varphi',\varphi'')+\log_2 K(\varphi,\varphi',\varphi'')\le \log_2\mathcal{M}(\varphi,\varphi',\varphi'')+4/r$.
Then we directly concatenate the three blocks of memory bits into a bit vector $M'=M(a,\varphi')+M(b,\varphi'')+M(\varphi,\varphi',\varphi'')$, and combine the spills into a superpill over the univser $K'= K(a,\varphi')\times K(b,\varphi'')\times K(\varphi,\varphi',\varphi'')$.
Since $\log_2 K'\le O(\log r)$ by our induction hypothesis, the superspill can be stored in constant number of cells if the word size $w = \Omega(\log r)$.
Given $a,b,\varphi,\varphi',\varphi''$ the query algorithm can easily locate the blocks of memory bits of the left child, the right child, and the tiny structure.
According to our induction hypothesis, we have
\[
  M'+\log_2 K'\le \log_2(\mathcal{N}(a,\varphi')\cdot\mathcal{N}(b,\varphi'')\cdot\mathcal{M}(\varphi,\varphi',\varphi''))+8\cdot\frac{n-2}{r}+4/r.
\]
Let $p(\cdot)$ be the distribution of $\varphi',\varphi''$ given $\varphi$, we have
\[
  p(\varphi',\varphi'')\triangleq \frac{\mathcal{N}(a,\varphi')\cdot\mathcal{N}(b,\varphi'')\cdot\mathcal{M}(\varphi,\varphi',\varphi'')}{\mathcal{N}(n,\varphi)}.
\]
We insert $p(\cdot)$ into our space upper bound, result in
\[
  \log_2\frac{1}{p(\varphi',\varphi'')}+M'+\log_2 K'\le \log_2\mathcal{N}(n,\varphi)+8\cdot\frac{n-2}{r}+4/r.
\]
By applying Lemma (\ref{compression lemma}) with the inequality above, we obtain a spill-over representation of subtree $v$ will spill universe $K_\star\le 2r$, and $M_\star$ memory bits, satisfying
\[
  M_\star+\log_2K_\star\le\log_2\mathcal{N}(n,\varphi)+8\cdot\frac{n-2}{r}+\frac{8}{r}=\log_2\mathcal{N}(n,\varphi)+8\cdot\frac{n-1}{r}.
\]
Note that at each step, we pack the labels of child nodes into a spill-over representation.
The remaining problem is how can the query algorithm know the label of root node.
To finish our constrution, we appy Lemma (\ref{compression lemma}) with $x=(\varphi,\varphi',\varphi'')$ at the root:
\[
  \log_2\frac{1}{p(\varphi,\varphi',\varphi'')}+M'+\log_2K'\le\log_2C_n+8\cdot\frac{n-1}{r}.
\]
Finally, we set $r\triangleq 8n$.

\paragraph{Final Query Algorithm.}
We state our final query algorithm in Algorithm \ref{algorithm: query}.

\begin{algorithm}
  \SetKwProg{Fn}{Function}{ is}{end}

  \caption{Final query algorithm to answer RMQ}
  \label{algorithm: query}
  \SetKwFunction{finda}{FindA'}
  \SetKwFunction{findb}{FindB'}
  \SetKwInOut{Input}{input}\SetKwInOut{Output}{output}
  \SetKwData{l}{l}
  \SetKwData{r}{r}
  \SetKwData{mid}{mid}
  \SetKwData{v}{CurrentNode}
  \SetKwData{lson}{LeftChild}
  \SetKwData{rson}{RightChild}
  \SetKwData{merge}{Merge}

  \Input{query $[a,b]$; the root \v.}
  \Output{index $i\in[n]$, $A[i]$ is the RMQ in range $[a,b]$.}
  \Begin
  {
	\If{$a=b$}{\Return $a$\;}
	\l$\gets1$,\r$\gets n$\;
	\mid$\gets\lceil(\l+\r)/2\rceil$\;
	Unpack child nodes \lson,\rson,\merge and $(\varphi_\text{\v},\varphi_\text{\lson},\varphi_\text{\rson})$ from the spill-over representation\;
	\While{$b\le\mid\lor a>\mid$}
	{
	  \eIf{$b\le\mid$}
	  {
		\r$\gets$\mid\;
		\v$\gets$\lson\;
	  }
	  {
		\l$\gets\mid+1$\;
		\v$\gets$\rson\;
	  }
	  Unpack child nodes \lson,\rson,\merge and $(\varphi_\text{\lson},\varphi_\text{\rson})$ from the spill-over representation\;
	}
	Find $(a',b')$, the lengths of the decreasing subsequence starting from $A[\mid]$ in subarray $A[\mid\dots a]$ and starting from $A[\mid+1]$ in sub array $A[\mid+1\dots b]$, respectively, with Algorithm \ref{relocate algorithm}\;
	Reduce the comparing $A[S_{r,\lson,a'}]$ and $A[S_{l,\rson,b'}]$ to \emph{select} problem with Eq(\ref{comparing by selecting})\;
	Answer the \emph{select} problem by \merge with the algorithm guaranteed by Claim \ref{select algorithm} in the universe calculated with Eq(\ref{def-M})\;
	\eIf{$A[S_{r,\lson,a'}]<A[S_{l,\rson,b'}]$}
	{
	  \Return \finda{\lson,$a'$}\;
	}
	{
	  \Return \findb{\rson,$b'$}\;
	}
  }
  \Fn{\finda{current node: $v$; target rank $a'$}}
  {
	\lIf{$L_v=R_v$}{\Return $L_v$}
	$\lson\gets\lchild{v}$\;
	$\rson\gets\rchild{v}$\;
	\lIf{$a'\le|S_{r,\rson}|$}{\Return \finda{\rson,$a'$}}
	\Return \finda{\lson,$a'-|S_{r,\rson}|+(|S_{r,\lson}|+|S_{r,\rson}|-|S_{r,v}|)$}\;
  }
  \Fn{\findb{current node: $v$; target rank $b'$}}
  {
	\lIf{$L_v=R_v$}{\Return $L_v$}
	$\lson\gets\lchild{v}$\;
	$\rson\gets\rchild{v}$\;
	\lIf{$b'\le|S_{l,\lson}|$}{\Return \finda{\lson,$b'$}}
	\Return \finda{\rson,$b'-|S_{l,\lson}|+(|S_{l,\lson}|+|S_{l,\rson}|-|S_{l,v}|)$}\;
  }
\end{algorithm}

\section{A Generalized Time-Redundancy Trade-Off}\label{app_trade_off}
The algorithm is a generalized version of the algorithm proposed by Fischer and Heun \cite{FH11}.
An exactly identical upper bounds was shown in \cite{NS14}.
The essential difference beween our algorithm and Navarro-Sadakane algorithm is the solution for \emph{findopen}:
our solution is much more clear and simple.
\begin{corollary}[generalized from Corollary 14 of \cite{FH11}]
Given the $\mathrm{DFUDS}$ of $\mathcal{M}_A$, $\mathrm{RMQ}_A(i,j)$ can be answered in $O(t)$ time with by the following sequence of operation $(1\le i\le j\le n)$.
\begin{enumerate}
  \item $x\gets\mathrm{select}_)(\mathrm{DFUDS},i+1)$
  \item $y\gets\mathrm{select}_)(\mathrm{DFUDS},j)$
  \item $w\gets\pm1\mathrm{RMQ}_E(x,y)$
  \item if $\mathrm{rank}_)(\mathrm{DFUDS},\mathrm{findopen}(\mathrm{DFUDS},w))=i$ then return $i$
  \item else return $\mathrm{rank}_)(\mathrm{DFUDS},w)$
\end{enumerate}
\end{corollary}
Our algorithm follows \mihai's schema \cite{Pat08}.
We break the parenthesis array of length $2n$ into $2n/r$ blocks of length $r$.
We choose $B\triangleq\frac{\log n}{\log r}$ and $t\triangleq\frac{\log r}{\log B}$.
For each block, we construct a segment tree with branching factor of $B$.
For any node with respect to range $[a,\dots,b]$ in the segment tree, we maintain
\begin{enumerate}
  \item the number of $)$'s in sub-array $\mathrm{DFUDS}[a',\dots,b']$, which is at most $r$;
  \item the minimum value in sub-array $E[a',\dots,b']$, which is in $[E[a]-r,E[a]+r]$.
\end{enumerate}
for all the $B$ sub-ranges $[a',\dots,b']$.
Note that we can encode them with at most $O(B\log r)=O(\log n)$ bits.

To answer select, rank and $\pm1$RMQ, we maintain three extra data structures $A,B,C$ for prefix sum array $N\in[n]^{n/r}$ and block minimum value array $M\in[n]^{n/r}$ : for any $i$, $N[i]$ is the total number of $)$'s in blocks $1,\dots,i$ of DFUDS; $M[i]$ is the minimum value in $i$-th block of array $E$.
\begin{enumerate}
  \item[$\mathrm{select}_)(i)$:]
	We do a predecessor search $i$ on $N$ with data structure $A$, and find the block $x$ which contains $i$-th $($. 
	To do the predecessor search, we adopt a variety of the algorithm from \cite{patracscu2006time}: in the leaf of the Van Emde Boas tree, recall that the leaf node denotes $\max\{v\in N:v< i\}$, we write down the index of the block $x$ such that $x=\max\{j\in[n/r]:N[j]< i\}$.
	Hence the block contains $i$-th $($ must be $x+1$, we then finish the query by querying on the segment tree with respect to block $x+1$.
  \item[$\mathrm{rank}_)(i)$:]
	The data structure $B$ is a copy of array $N$.
	To answer the prefix sum query $i$, we find the block $x$ which contains $i$, return $N[x-1]$ plus the answer of a rank query on the segment tree with respect to block $x$.
  \item[$\pm1$RMQ$(x,y)$:]
	We break the range $[x,y]$ into two in-block ranges and one out-block range.
	The two in-block RMQ can be easily answered by querying on at most two segment trees with respect to the two blocks which contains $x$ and $y$ respectively.
	Our data structure $C$ uses the algorithm from \cite{Sada07} as a black-box, which is a linear space data structure which can solve the out-block RMQ on array $M$ in constant time. 
	Finally, we return the minimum value among the three answers.
\end{enumerate}
For now, we can answer three of the four kinds of queries in $O(t)$ times with redundancy $n/(\frac{\log n}{t})^{O(t)}$.
To finish our proof, we take care of the last query \emph{findopen}$(w)$ now.

Recall that $E[i]\triangleq\mathrm{rank}_((i)-\mathrm{rank}_)(i)$.
For a closing parenthesis with index $i$, the index of its open parenthesis must be $\max\{j:E[j-1]=E[i]\}$.
Note that $\forall i, |E[i+1]-E[i]|=1$.
Which implies the index of the open parenthesis is the predecessor of $i$ in set $\{j\in[n]:E[j]\le E[i]\}$ plus one.
Thus the open parenthesis is in the same block as long as $\min E[\dots,i-1]$ ( i.e. the minimum value of the prefix of the sub-array) is not larger than $E[i]$ and we can find its index in $2\frac{\log r}{\log B}=2t$ time, since the minimum values of any sub-range are stored in the nodes of the segment tree.

To locate the block which contains the open parenthesis, we adopt the idea from Lemma 1 of \cite{munro2001succinct}.
A closing parenthesis is called \emph{far} if its matching parenthesis is located in a different block.
A far parenthesis is call a \emph{pioneer} if its matching parenthesis is located in a different block than its immediately next far parenthesis.
Obviously, if a far closing parenthesis is not a pioneer, its matching parenthesis and the matching parenthesis of the immediately next pioneer of the closing parenthesis must be located in the same block. 
Given a closing parenthesis, we check whether it is a far parenthesis by looking for its matching parenthesis in at most two blocks with $O(t)$ time.
We do a predecessor search to check whether a closing parenthesis is pioneer and find the immediately next pioneer of the closing parenthesis.
Jacobson\cite{jacobson1989space} noted that there are at most $4n/r-3$ pioneers if there are $2n/r$ blocks.
We have at most $O(n/r)$ values from a universe of size $2n$, so the second branch of \cite{patracscu2006time} can support query time $O(t)$ using space $(n/r)\cdot r^{\Omega(1/t)}\le n/B^{t-1}$ words.
To locate the block the matching parenthesis located in, we store the number of block in the leaf of the van Emde Boas tree.

To summarize, we support all the four kinds of queries with a redundancy of $n/B^{O(t)}=n/(\frac{\log n}{t})^{O(t)}$ bits and a time complexity of $O(t)$.


\end{document}